\documentclass[12pt]{article}

\usepackage{geometry}
\usepackage{yfonts}
\usepackage{subcaption}

\usepackage{amsthm}
\RequirePackage{etex}

\newcounter{mycou}

\usepackage{pict2e}

\makeatletter
\newcommand{\CCOp}{\mathord{\mathpalette\nicoud@YESNO{\nicoud@path{\fillpath}}}}
\newcommand{\nicoud@YESNO}[2]{%
  \begingroup
  \settoheight{\unitlength}{$#1X$}%
  \begin{picture}(0.7,1)
  \linethickness{\variable@rule{#1}}%
  \roundcap\roundjoin
  \nicoud@path{\strokepath}
  #2
  \Line(0.35,-0.35)(0.35,1.08)
  \Line(0.55,-0.35)(0.55,1.08)
  \end{picture}%
  \endgroup
}
\newcommand{\nicoud@path}[1]{%
  \moveto(0.5,0.9)
  \lineto(0.5,1)\lineto(0.6,1)\lineto(0.6,0.9)
  \closepath
  \moveto(0.3,0.9)
  \lineto(0.3,1)\lineto(0.4,1)\lineto(0.4,0.9)
  \closepath
  \moveto(0.5,-0.27)
  \lineto(0.5,-0.17)\lineto(0.6,-0.17)\lineto(0.6,-0.27)
  \closepath
  \moveto(0.3,-0.27)
  \lineto(0.3,-0.17)\lineto(0.4,-0.17)\lineto(0.4,-0.27)
  \closepath
  #1
}
\newcommand{\variable@rule}[1]{%
  \fontdimen8  
  \ifx#1\displaystyle\textfont3\else
    \ifx#1\textstyle\textfont3\else
      \ifx#1\scriptstyle\scriptfont3\else
        \scriptscriptfont3\relax
  \fi\fi\fi
}
\makeatletter

\makeatletter
\newcommand{\subalign}[1]{%
  \vcenter{%
    \Let@ \restore@math@cr \default@tag
    \baselineskip\fontdimen10 \scriptfont\tw@
    \advance\baselineskip\fontdimen12 \scriptfont\tw@
    \lineskip\thr@@\fontdimen8 \scriptfont\thr@@
    \lineskiplimit\lineskip
    \ialign{\hfil$\m@th\scriptstyle##$&$\m@th\scriptstyle{}##$\hfil\crcr
      #1\crcr
    }%
  }%
}
\makeatother

 \newtheoremstyle{theoremdd}
  {0}
  {0}
  {\itshape}
  {0pt}
  {\bfseries}
  {}
  { }
  {\thmname{#1}\thmnumber{ #2}\textnormal{\thmnote{ (#3)}}}

\theoremstyle{theoremdd}

\usepackage{setspace,multirow}

\makeatletter
\newcommand*{\rom}[1]{\expandafter\@slowromancap\romannumeral #1@}
\makeatother

\usepackage{lmodern,enumerate,rotating,anyfontsize}
\usepackage{amsmath,amsfonts,amssymb,mathtools,blkarray,soul,etoolbox}
\usepackage{tikz-cd}
\usepackage{tikz-network}

\makeatletter
\newcommand{\mylabel}[2]{#2\def\@currentlabel{#2}\label{#1}}
\makeatother

\usetikzlibrary{cd}
\usetikzlibrary{shapes,shapes.geometric,arrows,fit,calc,positioning,automata}
\usetikzlibrary {circuits.logic.IEC}

\usepackage{stmaryrd}
\SetSymbolFont{stmry}{bold}{U}{stmry}{m}{n}

\newcommand{\red}{\color{red}}
\newcommand{\blue}{\color{blue}}
\definecolor{green}{rgb}{0.1,0.7,0.1}

\DeclareMathOperator{\CC}{CC}

\DeclareMathOperator{\Obs}{Obs}
\DeclareMathOperator{\Det}{Det}

\DeclareMathOperator{\spec}{spec}

\DeclareMathOperator{\obs}{obs}

\DeclareMathOperator{\CriObs}{CriObs}
\DeclareMathOperator{\Pow}{Pow}
\DeclareMathOperator{\NS}{NS}

\allowdisplaybreaks

\newcommand{\Z}{\mathbb{Z}}

\newcommand{\N}{\mathbb{N}}

\newcommand{\dt}{\delta}

\newcommand{\ep}{\epsilon}
\newcommand{\vep}{\varepsilon}

\newcommand{\Scal}{\mathcal{S}}

\newcommand{\Sig}{\Sigma}
\newcommand{\s}{\sigma}

\newcommand{\Mt}{\mathcal{M}}

\newcommand{\llb}{\llbracket}
\newcommand{\rrb}{\rrbracket}

\newcommand{\Intr}{\mathsf{Intr}}
\newcommand{\Usr}{\mathsf{Usr}}
\newcommand{\QS}{{\red Q_S}}

\newcommand{\pred}{\mathsf{PRED}}

\newcommand{\QSordtwo}{{\red Q_S^{\mathsf{ord-2}}}}

\newtheorem{theorem}{Theorem}[section]
\newtheorem{definition}{Definition}

\newtheorem{remark}{Remark}

\newtheorem{fact}{Fact}

\usepackage{times,amssymb,amsfonts,mathtools,graphicx,tikz,algorithm,algorithmic,url,hhline,booktabs}
\usetikzlibrary{automata,arrows,positioning}

\usepackage[english]{babel}
\addto\captionsenglish{%
}
\addto\extrasenglish{%
}

\usepackage{aliascnt}

\theoremstyle{plain}

\newaliascnt{lemma}{theorem}
\newtheorem{lemma}[lemma]{Lemma}
\aliascntresetthe{lemma}

\newaliascnt{proposition}{theorem}

\aliascntresetthe{proposition}

\newaliascnt{corollary}{theorem}

\aliascntresetthe{corollary}

\newaliascnt{example}{theorem}
\newtheorem{example}[example]{Example}
\aliascntresetthe{example}

\usetikzlibrary{shadows}

\makeatletter
\let\NAT@parse\undefined
\makeatother
\usepackage[colorlinks,linkcolor=blue,anchorcolor=blue,citecolor=green,urlcolor=cyan,hyperindex]{hyperref}

\newcommand{\PSPACE}{\mathsf{PSPACE}}

\newcommand{\EXPTIME}{\mathsf{EXPTIME}}

\newcommand{\EXPSPACE}{\mathsf{EXPSPACE}}

\usetikzlibrary{automata,arrows,positioning}
\usetikzlibrary{petri}

\tikzset{elliptic state/.style={draw, ellipse, thick, fill=gray!10}
}

\tikzset{rectangular state/.style={draw, rectangle, thick, fill=gray!10}
}

\tikzset{emptystate/.style={}
}

\tikzset{
    partial ellipse/.style args={#1:#2:#3}{
        insert path={+ (#1:#3) arc (#1:#2:#3)}
    }
}

\tikzset{
node distance=3cm, 
every state/.style={thick, fill=gray!10}, 
initial text=$ $, 
}



\usepackage{tcolorbox}
\usepackage{tabularx}
\usepackage{array}
\usepackage{colortbl}
\tcbuselibrary{skins}

\title{High-order estimation-based properties and high-order observers for labeled finite-state automata\\
--- {\large Formulating and computing ``know what'' in a finite sequence of labeled finite-state automata}}

\author{Kuize Zhang\\
{\small School of Mathematics and Statistics}\\
{\small Xi'an Jiaotong University, 810049 Xi'an, China}\\
{\small kuize.zhang@xjtu.edu.cn}
\and
Xiaoguang Han\\
{\small College of Electronic Information and Automation}\\
{\small Tianjin University of Science and Technology, Tianjin 300222, China}\\
{\small hxg-allen@163.com}
\and
Alessandro Giua\\
{\small Department of Electrical and Electronic Engineering}\\
{\small University of Cagliari, 09123 Cagliari, Italy}\\
{\small giua@unica.it}
\and
Carla Seatzu\\
{\small Department of Electrical and Electronic Engineering}\\
{\small University of Cagliari, 09123 Cagliari, Italy}\\
{\small carla.seatzu@unica.it}
}

\begin{document}

\date{}

\maketitle

{\bf Abstract}
  In this paper, we consider labeled finite-state automata (LFSAs), extend some state estimation-based 
  properties from a single agent to a finite ordered set of agents. We also extend the notion of observer 
  to \emph{high-order observer} using our \emph{concurrent composition}.
  As a result, a general framework for characterizing high-order estimation-based properties is built, 
  in which each agent infers its preceding agent's estimation via all agents in front.  
  The high-order observer plays the role of a basic tool to verify such properties.

  In more detail, in our general framework, the system's structure is publicly known to all agents $A_1,\dots,A_n$;
  each agent $A_i$ has its own observable event set $E_i$, and additionally knows all its preceding agents' observable
  events but can only observe its own observable events. The intuitive meaning of our high-order observer is
  to characterize what 
  agent $A_n$ knows about what $A_{n-1}$ knows about \dots what $A_2$ knows about $A_1$'s state estimate of the system. 
  This general framework can be regarded as an automata representation of dynamic epistemic logic. Compared with
  the classical representation of dynamic epistemic logic based on fragments of logic, our representation has 
  advantages in property verification and flexibly changing agents to enforce properties.
  As case studies, this general framework applies to basic properties such as current-state opacity, strong current-state
  opacity, regular-language-based opacity, critical observability, 
  high-order opacity, etc. Special cases for which verification can be done more efficiently are also discussed.


\tableofcontents

\section{Introduction}
\label{sec:intro}

In this paper, we consider the scenario consisting of a \emph{finite-state automaton} (FSA) $G$ and
\emph{a finite sequence of $n$ agents} $A_i$, $1\le i \le n$.
The FSA is publicly known to all agents, each agent $A_i$ can observe a subset $E_i$ of events
of the FSA via a labeling function $\ell_i: E_i\to \Sig_i$, and each agent knows its preceding agents' 
observable events and labeling functions. We formulate and compute what $A_n$ knows 
about what $A_{n-1}$ knows about \dots what $A_2$ knows about $A_1$'s state estimate of $G$.

Take $n=3$ for example. Consider $G$ as an operating system, a user $\Usr$ who is working
on $G$, and an intruder $\Intr$ who wants to attack $G$ if $\Intr$ knows that $\Usr$ can
uniquely determine the state of the system. In this sense, $G$ is considered to be 
sufficiently safe for $\Usr$ to work on if $\Intr$ cannot know whether $\Usr$ can 
uniquely determine the state of $G$, and if $\Usr$ knows this. As a summary,
system $G$ is sufficiently safe if $\Usr$ knows that $\Intr$ cannot know whether $\Usr$
can uniquely determine the state of the system. Roughly speaking, ``I ($\Usr$) know you ($\Intr$) don't know whether I know''.

The above scenario is clearly in the category of ``epistemic logic'' with three ordered
agents. The difference compared with plenty of relevant results in the
literature on epistemic logic mainly lies in that, plenty of results were represented
as fragments of logic \cite{DynamicEpistemicLogic2008,Li2024Higher-orderEpistemicLogic,Muise2022EfficientMultiAgentEpistemicPlanning,Baltag2025TopologyOfSurprise}, 
while there are relatively much fewer characterizations of epistemic logic 
based on models (e.g., as transition systems \cite{Li2024Higher-orderEpistemicLogic}).
Characterizations of epistemic logic based on fragments of logic are very 
intuitive and show advantages in flexibly adjusting the physical meanings of the 
fragments.
In addition, all kinds of model checking techniques can be used to verify such 
fragments when they are decidable by converting the fragments to automata (or finite-transition systems).
However, it will be more convenient to directly use automata-based characterizations
to characterize epistemic logic in order to do verification, because automata are themselves models;
in addition, it will be more flexible to change behaviors of agents by modifying 
partial structures of automata. Therefore, fragments-based representations and 
automata-based representations for epistemic logic have their own advantages.

As mentioned above, in this paper, we will formulate and compute what $A_n$ knows 
about what $A_{n-1}$ knows about \dots what $A_2$ knows about $A_1$'s state estimate of $G$. 
We will propose a method called \emph{high-order observer} to compute this in $n$-$\EXPTIME$. That is, although the syntax for each agent appears not complex, the considered whole structure is complex.
As a comparison, when the transition systems are finite, the fragments of epistemic logic 
considered in \cite{Li2024Higher-orderEpistemicLogic} can be verified in polynomial time,
and the fragments considered in \cite{Lutz2006ComplexityPublicAnnouncementLogic} can be 
verified in $\EXPTIME$. However, there are no obvious containment relations between our structure and the fragments considered in these two papers. Further exploration on such studies will be of 
particular interest and left for future study.

In our paper, we will use our high-order observer to verify \emph{state-estimation-based properties} for LFSAs.
The study on \emph{state-estimation-based properties} of partially-observed dynamical systems dates back to
the 1950s \cite{Moore1956} in computer science and the 1960s \cite{Kalman1963MathDescriptionofLDS} in control science.
For over half a century, such properties have been central properties in these two areas.
In \emph{partially-observed discrete-event systems} (DESs for short) modeled by \emph{labeled finite-state automata}
(LFSAs), such properties have also been studied for over 30 years 
\cite{Ramadge1986ObservabilityDES,Shu2007Detectability_DES,Hadjicostis2020DESbook}.
Many results are based on a single agent who knows the structure of a system and can observe a subset
of events and estimate the system's state based on the system's structure and the agent's observation to the system.
The main tool used to do state estimation and verify estimation-based properties is called \emph{observer} which
is the powerset construction originally proposed by Rabin and Scott in 1959 \cite{RabinScott1959PowersetConstruction},
used to determinize a nondeterministic
finite automaton with $\vep$-transitions \cite{RabinScott1959PowersetConstruction,Shu2007Detectability_DES}.

We formulate \emph{a general framework} to study estimation-based properties for LFSAs from a high-level perspective.
We also develop a general approach to design a \emph{high-order observer} for verifying such properties based on two basic tools
--- \emph{observer} \cite{RabinScott1959PowersetConstruction,Shu2007Detectability_DES} and \emph{concurrent composition} \cite{Zhang2020DetPNFA}. 
We then show that many properties which have been previously presented in 
the literature, can be classified into our general framework as special cases.

An overview of the general framework is as follows.
Consider a finite-state automaton (FSA) $G$ whose structure is publicly known.
The development of the framework will be done in three steps.

The first step (shown in \autoref{sec:Order1Property}) considers a single agent $A_1$ 
who can observe a subset $E_1$ of events of $G$, and the properties to be formulated are based on $A_1$'s
state estimate of $G$. In this step, the properties are called \emph{of order-$1$}.
A number of properties in the literature are order-$1$ properties, e.g.,
current-state opacity \cite{Cassez2009DynamicOpcity,Saboori2007CurrentStateOpacity},
critical observability
\cite{Pola2017DecenCriticalObserver}, regular-language-based opacity \cite{Lin2011OpacityDES},
and strong current-state opacity \cite{Han2023StrongCSOpacity_DES}.

In a second step (shown in \autoref{sec:Order2Property}), we consider two agents 
$A_1$ and $A_2$, where $A_i$ can observe a subset $E_i$ of events of $G$, $i=1,2$,
and we additionally assume that $A_2$ knows $E_1$ but cannot observe events in $E_1$ but not in $E_2$.
The properties to be formulated 
are based on $A_2$'s inference of $A_1$'s state estimate of $G$, that is, what
$A_2$ knows about $A_1$'s state estimate of $G$. In this step, the properties 
are called \emph{of order-$2$}. For example, the high-order opacity studied in 
\cite{Cui2022YouDontKnowWhatIKnow} is a special case of our order-$2$ current-state opacity ---
one of our order-$2$ properties (see Appendix~\ref{appendix}).

In a third step (shown in \autoref{sec:Order3Property}), we further generalize this setting by considering
a finite number of ordered agents $A_1,\dots,A_n$. 
The properties to be formulated are based on what $A_n$ knows
about what $A_{n-1}$ knows about \dots what $A_2$ knows about $A_1$'s state estimate
of $G$. In this step, the properties are called \emph{of order-$n$}.

We define an order-$n$ observer to verify an order-$n$ estimation-based property.
An order-$n$ observer can be computed in 
$n$-$\EXPTIME$, and hence an order-$n$ estimation-based property can be verified in $n$-$\EXPTIME$.
We also study special cases in which verification can be done more efficiently, e.g., when for all agents
$i$, $E_i$ and $E_{i+1}$ have containment relations, and the labeling functions of all agents are 
projections, verification can be done in $\EXPTIME$.

\section{Preliminaries}
\label{sec:prelim}

{\bf Notation.} 
Symbols $\N$ and $\Z_+$ denote the set of nonnegative integers and the set of positive integers, respectively.
For two nonnegative integers $i\le j$, $\llb i,j\rrb$ denotes the set of all integers no less than $i$ and no greater 
than $j$.
For a set $S$, $|S|$ denotes its cardinality and $2^S$ its power set. Symbols $\subset$,
$\not\subset$, and $\subsetneq$ denote the \emph{subset of}, \emph{not a subset of}, and
\emph{a strict subset of} relations, respectively. For two sets $A$ and $B$, $A\setminus B$ denotes
$\{x\in A|x\notin B\}$.
Let $\Sigma$ denote a finite \emph{alphabet}. 
Elements of $\Sigma$ are called \emph{letters}. As usual, 
$\Sig^*$ denotes the set of \emph{words} or \emph{strings} (i.e., finite-length sequences of letters)
over $\Sig$ including the empty word $\epsilon$, 
and denote $\Sig^{+}:=\Sig^*\setminus\{\epsilon\}$.
The \emph{length} of a word $w\in\Sig^*$ is the number of letters, counting repetitions, occurring in $w$,
and is denoted by $|w|$.
A \emph{(formal) language} is a subset of $\Sig^*$.
For two languages $L_1,L_2\subset\Sig^*$, their \emph{concatenation} $L_1L_2$ is defined as $\{e_1e_2|e_1\in L_1,
e_2\in L_2\}$. For a function $f:A\to B$ and a subset $A'\subset A$, $f|_{A'}$ denotes the restriction of $f$
to $A'$.

\begin{definition}
  A \emph{finite-state automaton} (FSA) is a quadruple 
  \begin{align}\label{FSA}
	G=(Q,E,\dt,Q_0),
  \end{align} where
  \begin{enumerate}
  	\item $Q$ is a finite set of \emph{states},
    \item $E$ is an alphabet of \emph{events},
	\item $\dt:Q\times E \to 2^Q$ is the \emph{transition function} (equivalently described as
		$\dt\subset Q\times E\times Q$ such that $(q,e,q')\in\dt$ if and only if $q'\in\dt(q,e)$),
	\item $Q_0\subset Q$ is a set of \emph{initial states}.
  \end{enumerate}
\end{definition}

Transition function $\dt$ is recursively extended to $Q\times E^*\to 2^Q$: for all $q\in Q$, $u\in E^*$, and $e\in E$,
$\dt(q,\ep)=\{q\}$, $\dt(q,ue)=\bigcup_{p\in\dt(q,u)}\dt(p,e)$. 
Automaton $G$ is called \emph{deterministic} if $Q_0=\{q_0\}$ for some $q_0\in Q$, and
for all $q\in Q$ and $e\in E$, $|\dt(q,e)|\le 1$. A deterministic FSA $G$ is also 
denoted as $G=(Q,E,\dt,q_0)$, and in this case, $\dt$ is a partial function from $Q\times E$ to $Q$.

A transition $q\xrightarrow[]{e} q'$ with $q'\in \dt(q,e)$ means that when $G$ is in state $q$ and event 
$e$ occurs, $G$ transitions to state $q'$. A sequence $q_0\xrightarrow[]{e_1}\cdots \xrightarrow[]{e_n}
q_n$ of consecutive transitions with $n\in\N$ is called a \emph{run}\footnote{When $n=0$, the run degenerates
to a single state $q_0$.},
in which the event sequence $e_1\dots e_n$ is called
a \emph{trace (generated by $G$)} if $q_0\in Q_0$. A state $q\in Q$ is \emph{reachable} if there is a run
from some initial state to $q$. The \emph{reachable part} of $G$ consists of all
reachable states and transitions between them.
When showing an automaton, we usually only show its reachable part.
The \emph{language} $L(G)$ \emph{generated by $G$} is the set of traces generated by $G$. 

Occurrences of events of an FSA $G$ may be observable or not. Let alphabet $\Sig$ denote the set
of \emph{labels/outputs}. The \emph{labeling function} is defined as $\ell:E\to \Sig\cup\{\ep\}$, and 
is recursively extended to $\ell:E^*\to \Sig^*$. Denote $E_o=\{e\in E|\ell(e)\in \Sigma\}$,
$E_{uo}=\{e\in E|\ell(e)=\ep \}$, where the former denotes the set of \emph{observable events} and the latter 
denotes the set of \emph{unobservable events}. When an observable $e$ occurs; its label $\ell(e)$ is observed,
when an unobservable event occurs, nothing is observed.


A \emph{labeled finite-state automaton} (LFSA) is denoted as 
\begin{equation}\label{LFSA}
  \Scal=(G,\Sig,\ell).
\end{equation}

The following definition of state estimate is critical to define all kinds of properties in DESs.
The \emph{current-state estimate} $\Mt(\Scal,\alpha)$ with respect to $\alpha\in\ell(L(G))$ is defined as
\begin{subequations}\label{eqn39:High-OrderOpacity}
\begin{align}
  &\Mt(\Scal,\alpha)=\Mt_{\ell}(G,\alpha) \\
  =& \{q\in Q|(\exists\text{ run }q_0\xrightarrow[]{s} q)[q_0\in Q_0\wedge \ell(s)=\alpha]\},
\end{align}
\end{subequations}
which is the set of states $G$ may be in when $\alpha$ is observed.

For a labeling function $\ell:E\to\Sig\cup\{\ep\}$, we also denote $\ell$ as $\ell_{E_o}$ to indicate 
the observable event set $E_o$. In particular, when the restriction of $\ell_{E_o}$ to $E_o$ is the identity,
we also call the labeling function a \emph{projection} and also write it as $P_{E_o}$.

In the sequel, we sometimes refer to an LFSA $\Scal$ as an FSA $G$ with respect to a labeling function
$\ell$.




We now recall our concurrent composition that was first proposed in 
\cite{Zhang2020DetPNFA} to verify the \emph{negation} of strong detectability in 
arbitrary LFSAs and later on used in \cite{Zhang2023PTIMEVerEnforDetDES,Zhang2023UnifiedFrame4DES,Miao2025StrongInitialFinalOpacity}, for the verification of a series of properties such as strong detectability, diagnosability, and opacity. 
The concurrent-composition method provides a unified and currently the most efficient method to verifying 
all kinds of inference-based properties such as strong detectability, diagnosability, and predictability in LFSAs
\cite{Zhang2023UnifiedFrame4DES}, while the classical widely-used methods --- the detector
method \cite{Shu2011GDetectabilityDES} for verifying strong detectability, the twin-plant method
\cite{Jiang2001PolyAlgorithmDiagnosabilityDES} and the verifier method 
\cite{Yoo2002DiagnosabiliyDESPTime,Genc2009PredictabilityDES} for verifying diagnosability and
predictability are not universal and
all depend on two fundamental assumptions: deadlock-freeness (there are no dead states) and
divergence-freeness (no infinite runs of unobservable events can be generated), because these old methods were
proposed to verify inference-based properties themselves.
Apart from verifying inference-based properties, when used in combination with an
observer\footnote{Originally
proposed in \cite{RabinScott1959PowersetConstruction} for determinizing nondeterministic finite
automata with $\varepsilon$-transitions. The terminology ``observer'' dates back to 
\cite{Ozveren1990ObservabilityDES,Shu2007Detectability_DES}.}, the concurrent composition provides
a unified and currently the most efficient method to verifying all kinds 
of standard versions of state-based opacity and strong versions of state-based opacity in LFSAs
\cite{Zhang2023UnifiedFrame4DES,Han2023StrongCSOpacity_DES}.
See \cite[Table~4 and Table~5]{Zhang2023UnifiedFrame4DES} for details.

\begin{definition}[\cite{Zhang2020DetPNFA}]\label{FA:def_CCa}
	Consider two LFSAs $\Scal^i=(Q_i,E_i,\dt_i,Q_{0i},\Sig,\ell_i)$, $i=1,2$, where $E_{io}$ and $E_{iuo}$ denote 
	the set of observable events of $\Scal^i$ and the set of unobservable events of $\Scal^i$, respectively.
	The \emph{concurrent
	composition} $\CC(\Scal^1,\Scal^2)$ 
	of $\Scal^1$ and $\Scal^2$ 
	is defined by LFSA
	\begin{equation}\label{FA:eqn8}
	  \CC(\Scal^1,\Scal^2) = (Q',E',\dt',Q_0',\Sig',\ell'),
	\end{equation} where
	\begin{enumerate}
	\item $Q'=Q_1\times Q_2$;
	\item $E'=E_o'\cup E_{uo}'$, where $E_o'=\{(e_1,e_2)|e_1\in E_{1o},e_2\in E_{2o},
		\ell_1(e_1)=\ell_2(e_2)\}$,
		$E_{uo}'=\{(e_1,\epsilon)|e_1\in E_{1uo}\}\cup
		\{(\epsilon,e_2)|e_2\in E_{2uo}\}$;
	  \item for all $(q_1,q_2),(q_3,q_4)\in Q'$, $(e_{1o},e_{2o})
		\in E_o'$, $(e_{1uo},\epsilon),(\epsilon,e_{2uo})\in E_{uo}'$,
		\begin{itemize}
		  \item $((q_1,q_2),(e_{1o},e_{2o}),(q_3,q_4))\in\dt'$ 
			if and only if $(q_1,e_{1o},q_3)\in\dt_1$, $(q_2,e_{2o},q_4)\in\dt_2$,
		  \item $((q_1,q_2),(e_{1uo},\epsilon),(q_3,q_4))\in\dt'$ 
			if and only if $(q_1,e_{1uo},q_3)\in\dt_1$, $q_2=q_4$,
		  \item $((q_1,q_2),(\epsilon,e_{2uo}),(q_3,q_4))\in\dt'$ 
			if and only if $q_1=q_3$, $(q_2,e_{2uo},q_4)\in\dt_2$;
		\end{itemize}
	\item $Q_0'=Q_{01}\times Q_{02}$;
	\item 
	  for all $(e_{1o},e_{2o})\in E_o'$, $(e_{1uo},\epsilon)\in E_{uo}'$, and
	  $(\epsilon,e_{2uo})\in E_{uo}'$, $\ell'((e_{1o},e_{2o})):=\ell_1(e_{1o})=\ell_2(e_{2o})$,
	  $\ell'((e_{1uo},\epsilon)):=\ell_1(e_{1uo})=\ep$, $\ell'((\epsilon,e_{2uo})):=\ell_2(e_{2uo})=\ep$.
	\end{enumerate}
\end{definition}

\begin{definition}[\cite{RabinScott1959PowersetConstruction,Shu2007Detectability_DES}]\label{FA:def_obs}
  Consider an LFSA $\Scal=(G,\Sigma,\ell)$, where $G$ is as in \eqref{FSA}.
  Its \emph{observer} $\Obs(\Scal)=\Obs_{\ell}(G)$\footnote{The term ``observer''
  dates back to \cite{Ozveren1990ObservabilityDES,Shu2007Detectability_DES}.} is defined by
  a deterministic FSA
	\begin{align}\label{FA:eqn_observer}
	  (Q_{\obs},\ell(E_o),\dt_{\obs},q_{0\obs}),
	\end{align}
	where
		\begin{enumerate}
			\item $Q_{\obs}=2^Q$,
			\item $\ell(E_o)=\ell(E)\setminus\{\ep\}$,
			\item for all $X\in Q_{\obs}$ and $a\in\ell(E_o)$, $\dt_{\obs}(X,a)= \bigcup_{q\in X}\bigcup_{
			  \substack{e\in E_o\\\ell(e)=a}}\bigcup_{s\in (E_{uo})^*}\dt(q,es)$,
			\item $q_{0\obs}=\bigcup_{q_0\in Q_0}\bigcup_{s\in (E_{uo})^*}\dt(q_0,s)$.
		\end{enumerate}
\end{definition}

The observer $\Obs(\Scal)$, actually the powerset construction 
\cite{RabinScott1959PowersetConstruction},
can be computed in time exponential in the size of $\Scal$, and has been extensively used 
for many years in both the computer science community and the control community.
The observer $\Obs(\Scal)$ aggregates state estimates of $\Scal$ along all generated
label sequences, which is formulated as follows:
\begin{fact}[\cite{Shu2007Detectability_DES}]\label{fact1:High-OrderOpacity}
  for every run $q_{0\obs}\xrightarrow[]{\alpha}X$ of $\Obs(\Scal)$, $X=\Mt(\Scal,\alpha)$.
\end{fact}

\begin{definition}[\cite{Shu2011GDetectabilityDES}]
	Consider an LFSA $\Scal=(G,\Sig,\ell)$ and its observer
	$\Obs(\Scal)$. The \emph{detector}\index{detector}
	$\Det(\Scal)$, also denoted as $\Det_{\ell}(G)$, of $\Scal$ is defined as a nondeterministic finite automaton
	$(Q_{\det},\ell(E_o),\dt_{\det},q_{0\det})$, where 
	\begin{enumerate}
	  \item $q_{0\det}=q_{0\obs}$,
	  \item $Q_{\det}=\{q_{0\det}\}\cup\{X\subset Q|1\le|X|\le 2\}$, 
	  \item for each state $X$ in $Q_{\det}$ and each label $\sigma\in\ell(E_o)$,
		  \begin{align*}
		  &\dt_{\det}(X,\s)=\\
		  &\left\{
		  \begin{array}[]{ll}
			\{X'|X'\subset\dt_{\obs}(X,\sigma),|X'|=2\} & \text{if }|\delta_{\obs}(X,\sigma)|\ge 2,\\
			\{\dt_{\obs}(X,\sigma)\} & \text{if }|\delta_{\obs}(X,\sigma)|= 1,\\
			\emptyset & \text{otherwise.}
		  \end{array}
		  \right.
  		  \end{align*}
	\end{enumerate}
\end{definition}

$\Det(\Scal)$ can be computed in time polynomial in the size of $\Scal$. A critical relation between $\Det(\Scal)$
and $\Obs(\Scal)$ is as follows. 

\begin{lemma}[{\cite[Proposition~3]{Zhang2023RemoveAssumptionsSPDetect}}]\label{lem1:High-OrderOpacity} 
  Consider an LFSA $(G,\Sigma,\ell)$. Consider a run $q_{0\obs}\xrightarrow[]{e_1}X_1\xrightarrow[]{e_2}\cdots
  \xrightarrow[]{e_n}X_n$ in observer $\Obs(G,\Sig,\ell)=(Q_{\obs},\ell(E_o),\dt_{\obs},q_{0\obs})$, where
  $e_1,\dots,e_n\in \Sigma$, $X_n\ne\emptyset$. Choose $X_n'\subset X_n$ satisfying $|X_n'|=2$ if $|X_n|
  \ge 2$, and $|X_n'|=1$ otherwise. Then there is a run $q_{0\obs}\xrightarrow[]{e_1}X_1'\xrightarrow[]{e_2}\cdots
  \xrightarrow[]{e_n}X_n'$ in detector $\Det(G,\Sig,\ell)=(Q_{\det},\ell(E_o),\dt_{\det},q_{0\det})$
  (note that $q_{0\obs}=q_{0\det}$), where $|X_i'|=2$ if $|X_i|\ge 2$,
  $i\in\llb 1,n-1\rrb$.
\end{lemma}

\autoref{lem1:High-OrderOpacity} implies that along a run of the observer of an
LFSA from the initial state to some state which is not equal to $\emptyset$,
one can construct a run of its detector from end to head under the same label sequence
such that the cardinality of each state can be $1$ or $2$: states of cardinality $1$ appear in the detector if they appear in the observer; states with cardinality $2$ may either originate from a state of the observer with cardinality $2$ or from a state of the observer with cardinality larger than $2$. \autoref{lem1:High-OrderOpacity}
plays a central role in getting an $\EXPTIME$ algorithm for verifying
the special order-$2$ property as shown in \autoref{subsubsec:Order2Opacity}.


\section{Order-\texorpdfstring{$1$}{1} estimation-based problems}
\label{sec:Order1Property}

This section formalizes a general framework of order-$1$ properties which contains a single agent.
Then, it is shown that a number of properties studied in the literature, including current-state opacity,
critical observability, regular-language-based opacity, and 
strong current-state opacity belong to
this framework.
In this section we also propose some verification approaches that do not provide an improvement to the literature in terms of computational efficiency. Their importance lies in the fact they can be generalized to multiple agents.

\subsection{The general framework}

Consider an FSA $G=(Q,E,\dt,Q_0)$, an agent $A_1$ with its set $E_1\subset E$ of observable events and its
labeling function $\ell_{E_1}:E\to\Sig_1\cup\{\ep\}$.
In the traditional estimation-based problems, agent $A_1$ knows the structure of $G$.

Denote
\begin{subequations}\label{eqn36:High-OrderOpacity}
  \begin{align}
	\ell_{E_1} &=: \ell_1, & \Obs_{\ell_{E_1}}(G) &=: \Obs_1(G),\\
	\Det_{\ell_{E_1}}(G) &=: \Det_1(G), & \Mt_{\ell_{E_1}}(G,\alpha) &=: \Mt_1(G,\alpha),
  \end{align}
\end{subequations}
for short.

Recall that with respect to an observation sequence $\alpha\in \ell_1(L(G))$,
$A_1$'s current-state estimate of $G$ is $\Mt_1(G,\alpha)$. Define a \emph{predicate
of order-$1$ as} 
\begin{align}\label{eqn:order1predicate}
  \pred^1 \subset 2^Q.
\end{align}
Note that the superscript $1$ in \eqref{eqn:order1predicate} means there is a unique agent,
while the subscript $1$ in \eqref{eqn36:High-OrderOpacity} refers to the first agent.
Then an order-$1$ estimation-based property is defined as follows.
\begin{definition}\label{def:order1SEproperty}
  An FSA $G$ satisfies the \emph{order-$1$ estimation-based 
  property $\pred^1$ with respect to agent $A_1$} if 
  \begin{equation}\label{eqn13:High-OrderOpacity}
	\{\Mt_1(G,\alpha)|\alpha\in \ell_1(L(G))\}\subset\pred^1.
  \end{equation}
\end{definition}

One can see that the observer is a basic tool which can be used to verify whether an FSA
$G$ satisfies this property with respect to agent $A_1$.

\begin{theorem}\label{thm6:High-OrderCSO} 
  An FSA $G$ satisfies the order-$1$ estimation-based property $\pred^1$ with respect  
  to agent $A_1$ if and only if in the observer $\Obs_1(G)$, every reachable state
  $X$ belongs to $\pred^1$.
\end{theorem}

\begin{proof}
  Directly follows from \autoref{fact1:High-OrderOpacity}.
\end{proof}

\autoref{thm6:High-OrderCSO} provides an exponential-time algorithm 
for verifying the order-$1$ estimation-based property $\pred^1$,
because it takes exponential time to compute the observer
$\Obs_1(G)$. Note that the reachable part of 
$\Obs_1(G)$ is sufficient to verify the property.
Particularly, if except for the initial state of 
$\Obs_1(G)$, all reachable states have cardinalities no greater than $2$, then the reachable part of
$\Obs_1(G)$ can be computed in polynomial time. 
In this particular case,
one can get a polynomial-time verification algorithm.

Consider a special type 
$\pred^1\subset 2^Q$ of predicates satisfying that for every $X\subset Q$, if 
$|X|>2$ then $X \not\in \pred^1$, the satisfiability of the order-$1$ 
estimation-based property $\pred^1$ 
can be 
verified in polynomial time, because this property does not hold if and only if
in $\Obs_1(G)$, there is a reachable state $X$ with $|X|>2$. One can check this condition by computing 
the reachable states of $\Obs_1(G)$ from the initial state one by one, once a reachable state with 
cardinality greater than $2$ (including the initial state) is computed, stop and return that the property does
not hold; otherwise,
the whole reachable part of $\Obs_1(G)$ will be computed in polynomial time and all reachable states
have cardinalities no greater than $2$.

Until now, we have not associated any physical meaning with an order-$1$ estimation-based
property. 
In the following, we recall a number of properties defined in the literature that can be 
classified as the order-$1$ estimation-based property.
In order to do classification for a static property, that is, a property that only involves which states are reachable finally (e.g., current-state opacity \cite{Cassez2009DynamicOpcity,Saboori2007CurrentStateOpacity} 
  and critical observability \cite{Pola2017DecenCriticalObserver}),
one can directly use the observer $\Obs(\Scal)$; but for a dynamic property, that is, a property
that involves not only the finally reachable states but also the system's evolution via which the states are reachable (e.g., strong current-state opacity \cite{Han2023StrongCSOpacity_DES}
and regular-language-based opacity \cite{Lin2011OpacityDES}),
one needs to first compute the concurrent composition of the system and another system containing the evolution information
and then compute the observer of the concurrent composition.

\subsection{Case study 1: Current-state opacity}
\label{subsec:CSO}

Specify a subset $\QS\subset Q$ of secret states. Current-state opacity 
\cite{Cassez2009DynamicOpcity,Saboori2007CurrentStateOpacity} means that,
whenever a secret state is visited at the end of a run, there is another run whose final state is 
not secret such that the two runs look the same to an intruder.
This property is static, because it only involves in the final states of the runs and not how they have been reached.
In this setting, agent $A_1$ is regarded as an intruder $\Intr$. 

\begin{definition}[\cite{Saboori2007CurrentStateOpacity}]\label{FA:def_CSO}
  Consider an LFSA $(G,\Sig_1,\ell_1)$ and a subset $\QS\subset Q$ of secret states.
  FSA $G$ is called \emph{current-state opaque} \emph{with respect to $\ell_1$
  and $\QS$}
  if for every run $q_0\xrightarrow[]{s}q$ with 
  $q_0\in Q_0$ and $q\in \QS$, there exists a run $q_0'\xrightarrow[]{s'}q'$ such that $q_0'\in Q_0$,
  $q'\in Q\setminus \QS$, and $\ell_1(s)=\ell_1(s')$.
\end{definition}

To show that current-state opacity is a particular order-$1$ estimation-based property,
we define the order-$1$ predicate
\begin{align}\label{eqn14:High-OrderOpacity} 
  \pred^1(G,\QS):=\{X\subset Q|X\not\subset \QS\}\subset 2^Q.
\end{align}
The notion of current-state opacity is reformulated as follows.
\begin{definition}\label{def9:High-OrderCSO}  
  An FSA $G$ is called \emph{current-state opaque with respect to $\ell_1$ and $\QS$} if
  $$\{\Mt_1(G,\alpha)|\alpha\in \ell_1(L(G))\}\subset\pred^1(G,\QS).$$
\end{definition}

\autoref{def9:High-OrderCSO} is very similar to the following equivalent \autoref{def7:High-OrderCSO}.

\begin{definition}[\cite{Cassez2009DynamicOpcity}]\label{def7:High-OrderCSO} 
  An FSA $G$ is called \emph{current-state opaque with respect to $\ell_1$ and $\QS$} if
  for all $\alpha\in \ell_1(L(G))$, $\Mt_1(G,\alpha)\not\subset \QS$.
\end{definition}


\autoref{thm6:High-OrderCSO} provides an exponential-time algorithm for verifying 
current-state opacity of LFSAs. Furthermore, the current-state opacity verification problem is $\PSPACE$-complete in 
LFSAs \cite{Cassez2009DynamicOpcity}.

\subsection{Case study 2: A generalized version of critical observability}

In this subsection, we characterize a generalized version of critical observability, which
implies that under every generated label sequence $\alpha$,
the current-state estimate is contained in exactly one partition cell of a given partition of the state set $Q$. 
This property is also static.
We show this critical observability is also an order-$1$ estimation-based property.

\begin{definition}\label{FA:def_nCD}
  An FSA $G$ is called \emph{critically observable with respect to ${\blue \ell}_1$ and partition\\
  $\{Q_{\CriObs}^1,\dots,Q_{\CriObs}^n\}$\footnote{That is, $Q_{\CriObs}^1,\dots,Q_{\CriObs}^n$
  are nonempty, pairwise disjoint, and satisfy $\bigcup_{i=1}^{n}Q_{\CriObs}^i=Q$.} of $Q$} if
  for every $\alpha\in \ell_1(L(G))$, $\Mt_1(G,\alpha)\subset Q_{\CriObs}^i$ for exactly one $i$ in $\llb 1,n \rrb$.
\end{definition}

Define a special type of predicates:
\begin{subequations}\label{eqn17:High-OrderOpacity}
\begin{align} 
  &\pred^1_{\CriObs}\\
  = & \{X\subset Q|X\subset Q_{\CriObs}^i\text{ for exactly one }i\in\llb 1,n \rrb\}\subset 2^Q
\end{align}
\end{subequations}
Then the critical observability can be reformulated as follows.

\begin{definition}\label{FA:def_CD_PRED}
  An FSA $G$ is called 
  \emph{critically observable with respect to $\ell_1$ and $\{Q_{\CriObs}^1,\dots,Q_{\CriObs}^n\}$} if
  \[\{\Mt_1(G,\alpha)|\alpha\in \ell_1(L(G))\} \subset \pred^1_{\CriObs}.\]
\end{definition}

When the above partition has exactly two partition cells (i.e., $n=2$), the critical observability degenerates to the critical observability studied in \cite{Pola2017DecenCriticalObserver}.


\subsection{Case study 3: Regular-language-based opacity}

Consider three alphabets $\Sig_1,\Sig_2,\Sig$ and two labeling functions $\ell_1:\Sig_1\to\Sig\cup\{\ep\}$ and
$\ell_2:\Sig_2\to\Sig\cup\{\ep\}$. A language $L_1\subset \Sig_1^*$ is called \emph{L-opaque
with respect to a language $L_2\subset\Sig_2^*$ and labeling functions
$\ell_1$, $\ell_2$} if $\ell_1(L_1)\subset \ell_2(L_2)$ \cite{Bryans2008OpacityTransitionSystems},
where \emph{L} is short for \emph{language}. L-opacity is undecidable \cite{Bryans2008OpacityTransitionSystems}.
In \cite{Lin2011OpacityDES}, a special type of L-opacity was studied: 
\begin{definition}[\cite{Lin2011OpacityDES}]\label{FA:def_RLO}
  A regular language $L_1$ is called \emph{L-opaque
  with respect to a regular language $L_2$ and labeling functions $\ell_1,\ell_2$} if
  $\ell_1(L_1)\subset \ell_2(L_2)$.
\end{definition}

Apparently, L-opacity is dynamic because it relies on traces which reflect partial evolution information of the automata generating the languages.

For regular language $L_i$, $i=1,2$, denote a nondeterministic finite automaton (NFA) recognizing 
$L_i$ as $G_i=(Q_i,\Sig_i,\dt_i,Q_{0i},F_i)$, where 
$F_i\subset Q_i$ is a set of final states,
that is, $L_i=\{w\in\Sig_i^*|\dt_i(q_{0i},w)\cap F_i\ne\emptyset\text{ for some }q_{0i}\in
Q_{0i}\}$.
Denote $\Scal_i=(G_i,\Sig,\ell_i)$, $i=1,2$. 
Next, we show that regular-language-based opacity
(RL-opacity) is an order-$1$ estimation-based property by developing a 
verification method based on the observer (\autoref{FA:def_obs}) and the concurrent composition
(\autoref{FA:def_CCa}).

Without loss of generality, we assume
\begin{equation*}
  \tag{D}\label{quoteD:High-OrderOpacity}
  \parbox{\dimexpr\linewidth-4em}{%
	\strut
	for every run $q_{01}\xrightarrow[]{s_1} q_1$ in $G_1$, there is a run $q_{02}\xrightarrow[]{s_2} q_2$
	in $G_2$ such that $\ell_1(s_1)=\ell_2(s_2)$, where $q_{0i}\in Q_{0i}$, $i=1,2$.
	\strut
  }
\end{equation*}
If $G_1$ and $G_2$ do not satisfy \eqref{quoteD:High-OrderOpacity},
we modify $G_2$ in polynomial time to make $G_1$ and the modification $G_2^{\diamond}$ of $G_2$ satisfy
\eqref{quoteD:High-OrderOpacity}, where $G_2^{\diamond}$ still recognizes $L_2$.
We update $G_2$ as follows: Add a fresh non-final state $\diamond$ into $G_2$, where $\diamond\not\in Q_1\cup Q_2$.
For every state $q_2\in Q_2$ and every $\sigma\in \ell_1(\Sig_1)$, if there is no transition starting at $q_2$ (resp., $\diamond$) the label of whose event is $\sigma$,
add a transition $q_2\xrightarrow[]{e_2}\diamond$ (resp., $\diamond\xrightarrow[]{e_2}\diamond$),
where $e_2$ belongs to $\Sig_1\cup\Sig_2$ and satisfies $\ell_2(e_2)=\sigma$ if such an $e_2$ exists, 
otherwise choose a fresh event $e_2$ not in $\Sig_1\cup\Sig_2$ and we define its label as $\sigma$. 
Denote the current update of $G_2$ by $G_2^{\diamond}$, and the update of labeling function $\ell_2$ by $\ell_2^{\diamond}$.

We have $G_2^{\diamond}$ still recognizes $L_2$ and $G_1$ and $G_2^{\diamond}$ satisfy \eqref{quoteD:High-OrderOpacity}.
We also have $\ell_1(L_1)\subset\ell_2(L_2)$ if and only if $\ell_1(L_1)\subset
\ell_2^{\diamond}(L_2)$, 
hence we can equivalently consider $G_1,\ell_1$ and $G_2^{\diamond},\ell_2^{\diamond}$.

The following \autoref{thm1:FA:SRLOpacity} provides a necessary and sufficient condition for
regular-language-based opacity. It also shows regular-language-based opacity is an order-1 property.

Denote $\Scal_2^{\diamond}=(G_2^{\diamond},\Sig,\ell_2^{\diamond})$.
To simplify the notation, in the following \autoref{thm1:FA:SRLOpacity} and its proof,
we omit $\diamond$ and use $G_2$ and $\ell_2$, and $\Scal_2$ to denote $G_2^{\diamond}$, 
$\ell_2^{\diamond}$, and $\Scal_2^{\diamond}$, respectively.

\begin{theorem}\label{thm1:FA:SRLOpacity}
  Consider two LFSAs $\Scal_i=(G_i,\Sig_i,\ell_i)$, $i=1,2$, where $G_i=(Q_i,\Sig_i,\dt_i,Q_{0i},F_i)$ is an NFA
  and recognizes regular language $L_i$. Without loss of generality, assume $G_1$ and $G_2$ satisfy
  \eqref{quoteD:High-OrderOpacity}.
  Regular language $L_1$ is L-opaque with respect to regular language $L_2$ and labeling functions 
  $\ell_1,\ell_2$ if and only if in $\Obs(\CC(\Scal_1,\Scal_2))$, in each reachable state $X$ of 
  $\Obs(\CC(\Scal_1,\Scal_2))$, if there is a pair $(q_1,q_2)$ such that $q_1\in F_1$, then
  there is another pair $(q_1,q_3)$ such that $q_3\in F_2$.
\end{theorem}

\begin{proof}
  ``only if'': Assume $L_1$ is L-opaque with respect to $L_2$, $\ell_1$, and $\ell_2$. Choose an 
  arbitrary reachable
  state $X$ of $\CC(\Scal_1,\Scal_2)$ that contains a pair $(q_1,q_2)$ with $q_1\in F_1$ and $q_2\not\in F_2$.
  Choose a run $X_0\xrightarrow[]{\alpha} X$ in $\Obs(\CC(\Scal_1,\Scal_2))$, where $X_0$ is the initial state.
  Then there is a run
  $q_{01}\xrightarrow[]{s_1} q_1$ in $\Scal_1$ such that $\ell_1(s_1)=\alpha$
  and $q_{01}$ is initial.
  By assumption, there is a word $s_2\in L_2$ such that
  $\ell_1(s_1)=\ell(s_2)$. This implies a run $q_{02}\xrightarrow[]{s_2} q_3$ in $\Scal_2$ such that $q_{02}$ is initial and $q_3\in F_2$.
  The two runs $q_{01}\xrightarrow[]{s_1} q_1$ and $q_{02}\xrightarrow[]{s_2} q_3$ form a run 
  $(q_{01},q_{02})\xrightarrow[]{s'} (q_1,q_3)$ of $\CC(\Scal_1,\Scal_2)$ such that $s'$ produces label sequence
  $\alpha$. Hence $(q_1,q_3)\in X$.

  ``if'': Choose an arbitrary $s_1\in L_1$. Then there is a run $q_{01}\xrightarrow[]{s_1} q_1$ of $\Scal_1$ such that $q_{01}$ is initial and $q_1\in F_1$.
  Choose a run $q_{02}\xrightarrow[]{s_2} q_2$ of $\Scal_2$ such that $\ell_1(s_1)=\ell_2(s_2)$ (by \eqref{quoteD:High-OrderOpacity} such a run exists) and $q_{02}$ 
  is initial. Assume $q_2\not\in F_2$.
  Then for the run $X_0\xrightarrow[]{\ell_1(s_1)}X$ in $\Obs(\CC(\Scal_1,\Scal_2))$, where $X_0$ is the initial state,
  we have $(q_1,q_2)\in X$. By assumption, there is $q_3\in F_2$ such that
  $(q_1,q_3)\in X$. Then there is a run $q_{02}'\xrightarrow[]{s_3} q_3$ in $\Scal_2$ such that $\ell_2(s_3)=
  \ell_1(s_1)$ and $q_{02}'$ is initial. That is, $L_1$ is L-opaque with respect to $L_2$, $\ell_1$, and $\ell_2$.
\end{proof}

\begin{example}\label{exam3:High-OrderCSO}
  Consider the regular languages $L_1$ and $L_2$ recognized by $G_1$ and $G_2$
  shown in \autoref{FA:fig45}, where the labeling functions are
  $\ell_1(e_1)=\ell_2(e_1)=a$, $\ell_1(e_2)=b$, $1_1,2_1,1_2$ are final states, an input arrow from nowhere 
  denotes an initial state.
  One can see $L_1$ is not L-opaque with respect to $L_2$, $\ell_1$, and $\ell_2$,
  because $e_2\in L_1$, $\ell_1(e_2)=b$, and there is no word in $L_2$ with label sequence equal to $b$. Next, we use \autoref{thm1:FA:SRLOpacity} to verify this result.

  Now $G_1$ and $G_2$ do not satisfy 
  \eqref{quoteD:High-OrderOpacity}. We modify $G_2$ to obtain its modification
  $G_2^{\diamond}$ as in \autoref{FA:fig47} so that $G_1$ and $G_2^{\diamond}$
  satisfy \eqref{quoteD:High-OrderOpacity}.
  By above argument, in order to check if $L_1$ is L-opaque with respect to $L_2$, $\ell_1$, and $\ell_2$, we equivalently check if $L_1$ is L-opaque with respect to $L_2$, $\ell_1$, and $\ell_2^{\diamond}$.
  \begin{figure}[!htbp]
	\centering
	\subcaptionbox{$\Scal_1=(G_1,\Sig,\ell_1)$.\label{FA:fig45_1}}{
	  \begin{tikzpicture}[>=stealth',shorten >=1pt,auto,node distance=3.0 cm, scale = 1.0, transform shape,
	>=stealth,inner sep=2pt]

		\tikzstyle{emptynode}=[inner sep=0,outer sep=0]

		\node[state,initial, initial where = left] (01) {$0_1$};
		\node[state, accepting] (21) [right of =01] {$2_1$};
		\node[state, accepting] (11) [above = 1cm of 21] {$1_1$};
		\node[state] (31) [right of =11] {$3_1$};
		\node[state] (41) [right of =21] {$4_1$};

		\path [->]
		(01) edge node [above, sloped] {$e_1(a)$} (11)
		(01) edge node [above, sloped] {$e_2(b)$} (21)
		(11) edge node [above, sloped] {$e_1(a)$} (31)
		(21) edge node [above, sloped] {$e_2(b)$} (41)
		;
	\end{tikzpicture}
  }\hspace{0.5cm}
	\subcaptionbox{$\Scal_2=(G_2,\Sig,\ell_2)$.\label{FA:fig45_2}}{
	  \begin{tikzpicture}[>=stealth',shorten >=1pt,auto,node distance=3.0 cm, scale = 1.0, transform shape,
	>=stealth,inner sep=2pt]
		\node[state,initial, initial where = left] (02) {$0_2$};
		\node[state, accepting] (12) [right of =02] {$1_2$};
		\node[state] (32) [right of =12] {$3_2$};

		\path [->]
		(02) edge node [above, sloped] {$e_1(a)$} (12)
		(12) edge node [above, sloped] {$e_1(a)$} (32)
		;
		\end{tikzpicture}
	}
	\caption{Illustrative automata in \autoref{exam3:High-OrderCSO}.}
	\label{FA:fig45}
  \end{figure}

  \begin{figure}[!htbp]
	\centering
	\subcaptionbox{$\Scal_2^{\diamond}=(G_2^{\diamond},\Sig,\ell_2^{\diamond})$.\label{FA:fig46_2}}{
	  \begin{tikzpicture}[>=stealth',shorten >=1pt,auto,node distance=3.0 cm, scale = 1.0, transform shape,
	>=stealth,inner sep=2pt]
		\node[state, initial] (02) {$0_2$};
		\node[state, accepting] (12) [right of =02] {$1_2$};
		\node[state] (32) [right of =12] {$3_2$};
		\node[state] (d) [below of =12] {$\diamond$};

		\path [->]
		(02) edge node [above, sloped] {$e_1(a)$} (12)
		(12) edge node [above, sloped] {$e_1(a)$} (32)

		(02) edge node [above, sloped] {$e_2(b)$} (d)
		(12) edge node [above, sloped] {$e_2(b)$} (d)
		(32) edge node [below, sloped] {$e_1(a),e_2(b)$} (d)

		(d) edge [loop right] node {$e_1(a),e_2(b)$} (d)
		;
		\end{tikzpicture}
	  }\hspace{0.5cm}
	  \subcaptionbox{Part of $\CC(\Scal_1,\Scal_2^{\diamond})$.\label{FA:fig47_1}}{
	  \begin{tikzpicture}[>=stealth',shorten >=1pt,auto,node distance=3.0 cm, scale = 1.0, transform shape,
	>=stealth,inner sep=2pt]
		\node[rectangular state, initial, initial where = above] (01-01') {$(0_1,0_2)$};
		\node[rectangular state] (21-d) [right of =01-01'] {$(2_1,\diamond)$};
		\node[rectangular state] (11-12) [left of =01-01'] {$(1_1,1_2)$};

		\path [->]
		(01-01') edge node [above, sloped] {$(e_2,e_2)$} (21-d)
		(01-01') edge node [above, sloped] {$(e_1,e_1)$} (11-12)
		;
		\end{tikzpicture}
	}

	\subcaptionbox{Part of $\Obs(\CC(\Scal_1,\Scal_2^{\diamond}))$.\label{FA:fig47_2}}{
	  \begin{tikzpicture}[>=stealth',shorten >=1pt,auto,node distance=3.0 cm, scale = 1.0, transform shape,
	>=stealth,inner sep=2pt]
		\node[rectangular state, initial, initial where = above] (01-01') {$\{(0_1,0_2)\}$};
		\node[rectangular state] (21-d) [right of =01-01'] {$\{(2_1,\diamond)\}$};
		\node[rectangular state] (11-12) [left of =01-01'] {$\{(1_1,1_2)\}$};

		\path [->]
		(01-01') edge node [above, sloped] {$b$} (21-d)
		(01-01') edge node [above, sloped] {$a$} (11-12)
		;

		\end{tikzpicture}
	}
	\caption{Illustrative automata in \autoref{exam3:High-OrderCSO}.}
	\label{FA:fig47}
  \end{figure}

We compute $\CC(\Scal_1,\Scal_2^{\diamond})$ and $\Obs(\CC(\Scal_1,\Scal_2^{\diamond}))$ as in \autoref{FA:fig47}.
  In $\Obs(\CC(\Scal_1,\Scal_2^{\diamond}))$, there is a state $\{(2_1,\diamond)\}$ that contains a pair $(2_1,\diamond)$ with $2_1$ a final state of $G_1$, but does not contain a pair $(2_1,q)$ with $q$ a final state of $G_2^{\diamond}$. By \autoref{thm1:FA:SRLOpacity}, we also have $L_1$ is not L-opaque with respect to $L_2$, $\ell_1$, and $\ell_2$.
\end{example}

The L-opacity verification method shown in \autoref{thm1:FA:SRLOpacity} is more concise than that designed in \cite{Lin2011OpacityDES}. 

We use \autoref{thm1:FA:SRLOpacity} to define a special type of predicates as
\begin{equation}\label{eqn23:High-OrderOpacity} 
\begin{split}
  &\pred^1(Q_1,Q_2,F_1,F_2)  \\
  := &
  \{X\subset Q_1\times Q_2| (\exists q_1\in F_1)(\exists q_2\in Q_2)[(q_1,q_2)\in X]\\
  &\qquad\qquad\qquad\ \implies (\exists q_3\in F_2)[(q_1,q_3)\in X]\} \\
  \subset & 2^{Q_1\times Q_2}. 
\end{split}
\end{equation}
Then, the notion of RL-opacity is reformulated as follows.
\begin{definition}\label{FA:def_SRLO_Pred}
  A regular language $L_1$ is called \emph{L-opaque with respect to a regular language $L_2$ and labeling functions
  $\ell_1,\ell_2$} if 
  \begin{align}
	(Q_1\times Q_2)_{\Obs} \subset\pred^1(Q_1,Q_2,F_1,F_2),
  \end{align}
  where for $i=1,2$, NFA $G_i=(Q_i,\Sig_i,\dt_i,Q_{0i},F_i)$ recognizes language $L_i$,
  $\Scal_i=(G_i,\Sig,\ell_i)$, $G_1$ and $G_2$ satisfy \eqref{quoteD:High-OrderOpacity}.
  In addition, $(Q_1\times Q_2)_{\Obs}$ denotes the set of reachable states of observer $\Obs(\CC(\Scal_1,\Scal_2))$.
\end{definition}

\subsection{Case study 4: Strong current-state opacity}
\label{subsec:SCSO}

Stronger than current-state opacity, strong current-state opacity \cite{Han2023StrongCSOpacity_DES}
not only can guarantee that an intruder cannot be sure whether the current state of
an FSA is secret, but also can guarantee that 
the intruder cannot be sure whether some secret state has been 
visited. A \emph{non-secret run} is a run that contains no secret states.

\begin{definition}[\cite{Han2023StrongCSOpacity_DES}]\label{FA:def_SCSO}
  Consider an LFSA $(G,\Sig_1,\ell_1)=:\Scal$ and a subset $\QS\subset Q$ of secret states.
  FSA $G$ is called \emph{strongly current-state opaque} \emph{with respect to $\ell_1$
  and $\QS$}, or more concisely, LFSA $\Scal$ is called \emph{strongly current-state opaque with respect to 
  $\QS$},
	if for every run $q_0\xrightarrow[]{s}q$ with $q_0\in Q_0$ and $q\in \QS$, there exists a non-secret run $q_0'\xrightarrow[]{s'}q'$ 
	such that $q_0'\in Q_0$ and $\ell_1(s)=\ell_1(s')$.
\end{definition}

Strong current-state opacity is dynamic, because in the second run, the final state
is reachable from a special run which contains no secret state, but not reachable from an 
arbitrary run. Based on this,
unlike current-state opacity, one cannot directly use the observer to verify
strong current-state opacity.

Recall $G=(Q,E,\dt,Q_0)$.
Construct $G_1'=(Q,E,\dt,Q_0,\QS)$ as an NFA, $\Scal_1'=(G_1',\Sig_1,\ell_1)$ as an LFSA.
Construct the non-secret sub-automaton $\Scal_{\NS}$ of $\Scal$ that is obtained from
$\Scal$ by removing all secret states and the corresponding transitions. Denote
$\Scal_{\NS}=(G_1'',\Sig_1,\ell_1'')$, where $G_1''=(Q_{\NS},E_{\NS},\dt_{\NS},Q_{0\NS},Q_{\NS})$, $\ell_1''=\ell_1|_{E_{\NS}}$. Denote the regular 
languages recognized by $G_1'$ and $G_1''$ by $L'$ and $L''$, respectively. Then
FSA $G$ is strongly current-state opaque with respect to $\ell_1$ and $\QS$ if and only if
language $L'$ is L-opaque with respect to languages $L''$ and labeling functions 
$\ell_1$ and $\ell_1''$. We have shown strong current-state opacity is a special type of RL-opacity, hence it is also an order-$1$ estimation property.

\begin{remark}
  After transforming strong current-state opacity to RL-opacity, one can use the method with time complexity $O(2^{n(n+1)}m)$ in \autoref{thm1:FA:SRLOpacity} to verify strong current-state opacity of FSA $G=(Q,E,\dt,Q_0)$ with respect to $\ell_1$ and secret state set $\QS$, this method is less efficient than the verification method with time complexity $O(n^22^{n'+1}m)$ given in \cite{Han2023StrongCSOpacity_DES}, where $n=|Q|$, $n'=|Q_{\NS}|\le n$, $m=|E|$.
\end{remark}

\section{Order-\texorpdfstring{$2$}{2} estimation-based problems}
\label{sec:Order2Property}

\subsection{The general framework}

Consider an FSA $G=(Q,E,\dt,Q_0)$, two agents $A_1$ and $A_2$ with observable event sets $E_1\subset E$
and $E_2\subset E$ and labeling functions $\ell_{E_1}$ and $\ell_{E_2}$.
As mentioned before, assume both agents know the structure of $G$,
also assume $A_2$ knows $E_1$ but cannot observe events of $E_1\setminus E_2$.
Denote
\begin{subequations}\label{eqn37:High-OrderOpacity} 
  \begin{align}
	\ell_{E_i} &=: \ell_i, & \Obs_{\ell_{E_i}}(G) &=: \Obs_i(G),\\
	\Det_{\ell_{E_i}}(G) &=: \Det_i(G), & \Mt_{\ell_{E_i}}(G,\alpha) &=: \Mt_i(G,\alpha),
  \end{align}
\end{subequations}
for short, $i=1,2$.
We formulate the order-$2$ estimation-based property as follows. Recall with respect to a label sequence
$\alpha\in \ell_1(L(G))$ generated by $G$ observed by agent $A_1$, $A_1$'s current-state estimate of $G$ 
is $\Mt_1(G,\alpha)$. Agent $A_2$ knows $E_1$, so $A_2$ can infer $A_1$'s current-state estimate of
$G$ from $A_2$'s own observations to $G$. 

For a label sequence $\alpha$ observed by $A_2$, the real generated event sequence can be 
any $s\in \ell_2^{-1}(\alpha)\cap L(G)$, so the observation of $A_1$ can be $\ell_1(s)$
for any such $s$,
and then the inference of $A_1$'s current-state estimate from $A_2$ can be $\Mt_1(G,\ell_1(s))$
for any such $s$.
Formally, when $A_2$ observes $\alpha\in \ell_2(L(G))$, all possible inferences of $A_1$'s current-state estimate of $G$
by $A_2$ are represented by the set
\begin{subequations}\label{eqn15:High-OrderOpacity}
\begin{align}
  & \Mt_{A_1\leftarrow A_2}(G,\alpha)\label{eqn15_1:High-OrderOpacity}\\
  := &\{\Mt_1(G,\ell_1(s))|s\in \ell_2^{-1}(\alpha)\cap L(G)\} \subset 2^Q.\label{eqn15_2:High-OrderOpacity}
\end{align}
\end{subequations}
By definition, $\Mt_{A_1\leftarrow A_2}(G,\alpha)$ must contain $A_1$'s real current-state estimate of $G$.

Define a \emph{predicate of order-$2$ as} 
\begin{align}\label{eqn:order2predicate}
  \pred^2 \subset 2^{2^Q}.
\end{align}
Then an order-$2$ estimation-based property is defined as follows.
\begin{definition}\label{def:order2SEproperty}
  An FSA $G$ satisfies the \emph{order-$2$ estimation-based
  property $\pred^2$ with respect to agents $A_1$ and $A_2$} if 
  \begin{equation}\label{eqn16:High-OrderOpacity} 
	\{\Mt_{A_1\leftarrow A_2}(G,\alpha)|\alpha\in \ell_2(L(G))\}\subset\pred^2.
  \end{equation}
\end{definition}

In order to derive an algorithm to verify \autoref{def:order2SEproperty}, we use two basic tools ---
concurrent composition (as in \autoref{FA:def_CCa}) and observer (as in \autoref{FA:def_obs})
to define a new tool --- \emph{order-$2$ observer}.
The following constructive definition provides a procedure to compute an order-$2$ observer.
\begin{definition}\label{def:order2observer}
  Consider FSA $G$ as in \eqref{FSA}, agents $A_1$ and $A_2$ with observable event sets $E_1\subset E$
  and $E_2\subset E$ and labeling functions $\ell_{E_1}$ and $\ell_{E_2}$.
  Denote LFSAs $(G,\Sig_i,\ell_i)$ by $G_{A_i}$ for short, $i=1,2$. 
\begin{enumerate}
  \item Compute the observer $\Obs_1(G)$ of $G_{A_1}$, and denote $\Obs_1(G)$ as $\Obs_{A_1}$ for short.
  \item Compute the concurrent composition $\CC(G_{A_1},\Obs_{A_1})$ of LFSA $G_{A_1}$ and its observer $\Obs_{A_1}$,
	replace each event $(e_1,e_2)$ by $e_1$, replace the labeling function of $\CC(G_{A_1},\Obs_{A_1})$
	by $\ell_2$, and denote the modification of $\CC(G_{A_1},\Obs_{A_1})$ by $\CC_{A_1\to A_2}^{G,\Obs}$
	which is an LFSA.
  \item Compute the observer $\Obs(\CC_{A_1\to A_2}^{G,\Obs})$ of $\CC_{A_1\to A_2}^{G,\Obs}$, and call
	$\Obs(\CC_{A_1\to A_2}^{G,\Obs})$ \emph{order-$2$ observer} and denote it as $\Obs_{A_1\leftarrow A_2}
	(G)$, where $A_1\leftarrow A_2$ intuitively describes the scenario $A_2$ infers $A_1$. 
\end{enumerate}
\end{definition}
See \autoref{fig1:order2SEproperty} for an illustration.
It takes doubly exponential time to compute an order-$2$ observer $\Obs_{A_1\leftarrow A_2}(G)$.

\begin{figure}[!htbp]
  \centering
  \begin{tikzpicture}[circuit logic IEC]
	\matrix[column sep=5mm, row sep=0.5cm]
	{
	  \node (SusrR) {$G_{A_1}$}; & & \node (ObsusrR) {$\Obs_{A_1}$};\\
	  & \node (CCusrR) {$\CC$}; & \\
	  & \node (P1to2R) {$\to \ell_2$}; & \\
	  & \node (ObsIntrR) {$\Obs$}; & \\
	};
	\draw (SusrR.south) -- ++(down:2.0mm) -| (CCusrR.north);
	\draw (ObsusrR.south) -- ++(down:2.0mm) -| (CCusrR.north);
	\draw (CCusrR.south) -- (P1to2R);
	\draw (P1to2R) -- (ObsIntrR.north);
  \end{tikzpicture}
  \caption{Sketch of the order-$2$ observer and $2$-$\EXPTIME$ verification structure for the order-$2$ estimation-based property.}
  \label{fig1:order2SEproperty}
\end{figure}

The next \autoref{thm7:High-OrderCSO} shows several fundamental properties of the order-$2$ observer
$\Obs_{A_1\leftarrow A_2}(G)$, 
and plays a fundamental role in verifying the order-$2$ estimation-based property.

\begin{theorem}\label{thm7:High-OrderCSO}
  Consider an FSA $G$ as in \eqref{FSA}, agents $A_1$ and $A_2$ with observable event sets $E_1\subset E$
  and $E_2\subset E$ and labeling functions $\ell_{E_1}$ and $\ell_{E_2}$,
  and the order-$2$ observer $\Obs_{A_1\leftarrow A_2}(G)$.
  \begin{enumerate}[(i)]
	\item\label{item1:SEproperty}
	  The initial state of 
	  $\Obs_{A_1\leftarrow A_2}(G)$ is of the
	  form $\{(q_{0,1},X_0),\dots,(q_{0,m},X_0)\}$, where 
	  $\{q_{0,1},\dots,q_{0,m}\}=X_0$, $X_0$ is the initial state of $\Obs_{A_1}$.
	\item\label{item2:SEproperty}
	  $L(G) = L(\CC_{A_1\to A_2}^{G,\Obs})$.
	\item\label{item3:SEproperty}
	  For every reachable state $\{(q_1,X_1),\dots,(q_n,X_n)\}$ of $\Obs_{A_1\leftarrow A_2}(G)$, 
	  $q_j\in X_j$, $j\in\llb1, n\rrb$.
	\item\label{item4:SEproperty}
	  For every run $C_0\xrightarrow[]{\alpha} \{(q_1,X_1),\dots,(q_n,X_n)\}$ of 
	  $\Obs_{A_1\leftarrow A_2}(G)$, where $C_0$ is the initial state, 
	  $\{X_1,\dots,X_n\}=\Mt_{A_1\leftarrow A_2}(G,\alpha)$.
  \end{enumerate}
\end{theorem}

\begin{proof}
  \eqref{item1:SEproperty}, \eqref{item2:SEproperty} and \eqref{item3:SEproperty} hold by definition.
  For \eqref{item4:SEproperty}, 
  the initial state of order-$2$ observer $\Obs(\CC_{A_1\to A_2}^{G,\Obs})$ is $\{(q_{0,1},X_0),\dots,(q_{0,m},X_0)\}$.
  Then $$\{(q_{0,1},X_0),\dots,(q_{0,m},X_0)\}\xrightarrow[]{\ell_2(s)}\{(q_1,X_1),\dots,(q_n,X_n)\}$$
  is an arbitrary run of $\Obs(\CC_{A_1\to A_2}^{G,\Obs})$, where $s\in L(G)$. Then it holds that\\
  $\{X_1,\dots,X_n\}=\Mt_{A_1\leftarrow A_2}(G,\ell_2(s))$.
\end{proof}

By \autoref{thm7:High-OrderCSO}, the following \autoref{thm1:order2SEproperty} holds.

\begin{theorem}\label{thm1:order2SEproperty}
  An FSA $G$ satisfies an order-$2$ estimation-based
  property $\pred^2$ \eqref{eqn:order2predicate}
  with respect to agents $A_1$ and $A_2$ 
  if and only if for every reachable state $\{(q_1,X_1),\dots,(q_n,X_n)\}$ of the order-$2$ observer
  $\Obs_{A_1\leftarrow A_2}(G)$, $\{X_1,\dots,X_n\}\in \pred^2$.
\end{theorem}

\begin{remark}
  \autoref{thm1:order2SEproperty} provides an algorithm for verifying the order-$2$ estimation-based property
  in doubly exponential time. 
\end{remark}

\begin{example}\label{exam1:High-OrderCSO}
  Consider FSA $G_{ \label{cou1}}$ as in \autoref{fig1:High-OrderCSO}. We consider $E_1^{ }=\{b,c,d\}$ and $E_2^{ }=\{a,b\}$. 
  We use \autoref{thm1:order2SEproperty} and follow the sketch shown in 
  \autoref{fig1:order2SEproperty} to verify the order-$2$ estimation-based property
  $\pred^2=\{\emptyset\not\in Y\subset 2^{\{0,1,2,3,4,5\}}| (\exists X\in Y)[|X|>1]\}\subset 2^{2^{\{0,1,2,3,4,5\}}}$
  of $G $ with respect to agents $A_1^{ }$, $A_2^{ }$. 
  The FSA $G $ with respect to agents $A_1^{ }$ and $A_2^{ }$ are denoted by
  $G^{ }_{A_1^{ }}$ and $G^{ }_{A_2^{ }}$, respectively.
  The observer $\Obs_{A_1^{ }}$ of $G^{ }_{A_1^{ }}$ is shown in \autoref{fig2:High-OrderCSO}.
  The concurrent composition $\CC(G_{A_1},\Obs_{A_1})$ is shown in 
  \autoref{fig3:High-OrderCSO}. The order-$2$ observer 
  $\Obs(\CC^{G ,\Obs}_{A_1^{ }\to A_2^{ }})=
  \Obs_{A_1^{ }\leftarrow A_2^{ }}(G )$
  is shown in \autoref{fig4:High-OrderCSO}.
  In \autoref{fig4:High-OrderCSO},
  there is a reachable state $\{(2,B)\}$ in which 
  $B$ is a singleton and there is no state of the form
  $(q,X)$ with $|X|>1$. By \autoref{thm1:order2SEproperty}, $G $ does not satisfy the order-$2$
  estimation-based property. 

  Directly by definition, for the run $0\xrightarrow[]{a}1\xrightarrow[]{b}2$,
  one has $\ell_2^{ }(ab)=\gamma\delta$, $(\ell_2^{ })^{-1}(\gamma\delta)=\{ab\}$, $\ell_1^{ }(ab)=\alpha$,
  $\Mt_1^{ }(G ,\alpha)=\{2\}$.
  Hence $A_2^{ }$ knows that $A_1^{ }$ uniquely determines the current state when observing $\gamma\delta$.
  Then we conclude that $G $ does not satisfy the order-$2$
  estimation-based property.
  \begin{figure}[!htbp]
  \centering
  \subcaptionbox{Automaton $G $ considered in \autoref{exam1:High-OrderCSO}, where $E_1^{ }=\{b,c,d\}$,
  {$\ell_1(b) = \ell_1(c) = \alpha$}, $\ell_1(d) = \beta$,
  $E_2^{ }=\{a,b\}$, {$\ell_2(a) = \gamma$}, $\ell_2(b) = \delta$.\label{fig1:High-OrderCSO}}{
  	  \begin{tikzpicture}[>=stealth',shorten >=1pt,auto,node distance=2.0 cm, scale = 1.0, transform shape,
	>=stealth,inner sep=2pt]

	\node[initial, initial where = left, state] (0) {$0$};
	\node[state] (1) [right of =0] {$1$};
	\node[state] (2) [right of =1] {$2$};
	\node[state] (3) [right of =2] {$3$};
	\node[state] (4) [above right of =3] {$4$};
	\node[state] (5) [below right of =3] {$5$};

	\path [->]
	(0) edge node [above, sloped] {$a$} (1)
	(1) edge node [above, sloped] {$b$} (2)
	(2) edge [bend left] node {$b$} (3)
	(3) edge [bend left] node {$c$} (2)
	(3) edge node [above, sloped] {$a$} (4)
	(3) edge node [above, sloped] {$a$} (5)
	(4) edge [loop right] node {$d$} (4)
	;

    \end{tikzpicture}
  }\vspace{0.5cm}

  \subcaptionbox{The observer $\Obs_{A_1^{ }}$ of automaton $G_{A_1}$. 
  \label{fig2:High-OrderCSO}}{
	\begin{tikzpicture}[>=stealth',shorten >=1pt,auto,node distance=2.5 cm, scale = 1.0, transform shape,
	>=stealth,inner sep=2pt]

	\node[initial, initial where = left, rectangular state] (01) {$\{0,1\}$};
	\node[rectangular state] (2) [right of =01] {$\{2\}$};
	\node[rectangular state] (345) [right of =2] {$\{3,4,5\}$};
	\node[rectangular state] (4) [right of =345] {$\{4\}$};

	\path [->]
	(01) edge node {$\alpha$} (2)
	(2) edge [bend left] node {$\alpha$} (345)
	(345) edge [bend left] node {$\alpha$} (2)
	(345) edge node {$\beta$} (4)
	(4) edge [loop right] node {$\beta$} (4)
	;

    \end{tikzpicture}
  }

  \subcaptionbox{The concurrent composition $\CC(G_{A_1},\Obs_{A_1})$,
	where $A=\{0,1\}$, $B=\{2\}$, $C=\{3,4,5\}$, $D=\{4\}$.\label{fig3:High-OrderCSO}}{
	\begin{tikzpicture}[>=stealth',shorten >=1pt,auto,node distance=2.5 cm, scale = 1.0, transform shape,
	>=stealth,inner sep=2pt]

	\node[initial, initial where = left, rectangular state] (0A) {$(0,A)$};
	\node[rectangular state] (1A) [below =1cm of 0A] {$(1,A)$};
	\node[rectangular state] (2B) [right of =1A] {$(2,B)$};
	\node[rectangular state] (3C) [right of =2B] {$(3,C)$};
	\node[rectangular state] (4C) [above right of =3C] {$(4,C)$};
	\node[rectangular state] (5C) [below right of =3C] {$(5,C)$};
	\node[rectangular state] (4D) [right of =4C] {$(4,D)$};

	\path [->]
	(0A) edge node {$(a,\ep)$} (1A)
	(1A) edge node {$(b,\alpha)$} (2B)
	(2B) edge [bend left] node {$(b,\alpha)$} (3C)
	(3C) edge [bend left] node {$(c,\alpha)$} (2B)
	(3C) edge node [sloped, above] {$(a,\ep)$} (4C)
	(3C) edge node [sloped, above] {$(a,\ep)$} (5C)
	(4C) edge node {$(d,\beta)$} (4D)
	(4D) edge [loop right] node {$(d,\beta)$} (4D)
	;

    \end{tikzpicture}
  }\vspace{0.5cm}

  \subcaptionbox{The order-$2$ observer 
	$\Obs_{A_1^{ }\leftarrow A_2^{ }}(G )$,
	where $A,B,C,D$ are as above.
 	\label{fig4:High-OrderCSO}}{
	\begin{tikzpicture}[>=stealth',shorten >=1pt,auto,node distance=2.5 cm, scale = 1.0, transform shape,
	>=stealth,inner sep=2pt]

	\node[initial, initial where = left, rectangular state] (0A) {$\{(0,A)\}$};
	\node[rectangular state] (1A) [below =1cm of 0A] {$\{(1,A)\}$};
	\node[rectangular state] (2B) [right of =1A] {$\{(2,B)\}$};
	\node[rectangular state] (2B3C) [right =1cm of 2B] {$\{(2,B),(3,C)\}$};
	\node[rectangular state] (4C5C4D) [right =1cm of 2B3C] {$\{(4,C),(5,C),(4,D)\}$};

	\path [->]
	(0A) edge node {$\gamma$} (1A)
	(1A) edge node {$\delta$} (2B)
	(2B) edge node {$\delta$} (2B3C)
	(2B3C) edge node {$\gamma$} (4C5C4D)
	(2B3C) edge [loop above] node {$\delta$} (2B3C)
	;

    \end{tikzpicture}
  }\vspace{0.5cm}

  \caption{illustrative automata in \autoref{exam1:High-OrderCSO}.}
  \end{figure}

\end{example}

\subsection{Special cases verifiable in exponential time}
\label{subsec:Order2SEproperty}

For several special cases, verification of an order-$2$ estimation-based
property can be done in exponential time.

\subsubsection{Special case 1 --- \texorpdfstring{$E_1\subset E_2$ and $\ell_1$ and $\ell_2$ are projections}{E1subsetE2}}

Consider $E_1\subset E_2$. Then agent $A_2$ knows more on $G$ than agent
$A_1$, and can know exactly $A_1$'s current-state estimate of $G$. Formally, 
in this case, for all $\alpha\in \ell_2(L(G))$,
\begin{subequations}
  \begin{align}
	\Mt_{A_1\leftarrow A_2}(G,\alpha) &= \{\Mt_1(G,\ell_1(\alpha))\},\\
	\Mt_2(G,\alpha) &\subset \Mt_1(G,\ell_1(\alpha)).
  \end{align}
\end{subequations}

This implies the following result. 

\begin{fact}\label{lem4:High-OrderOpacity}
  Assume $E_1\subset E_2$ and $\ell_1$ and $\ell_2$ are projections. Then each state of order-$2$ observer
  $\Obs_{A_1\leftarrow A_2}(G)$ is of the form $\{(q_1,X),\dots,
(q_n,X)\}$ 
with $q_1,\dots,q_n\in X\subset Q$, $\Obs_{A_1\leftarrow A_2}(G)$
and $\Obs_{A_2}$ have the same number of reachable states. For every run $Y_0\xrightarrow[]{\alpha} 
\{(q_1,X),\dots, (q_n,X)\}$ of $\Obs_{A_1\leftarrow A_2}(G)$
and every run $X_0\xrightarrow[]{\alpha} X'$ of $\Obs_{A_2}$, where $Y_0$ and 
$X_0$ are the initial states of the two observers, $X'\subset X$.
\end{fact}

\begin{remark}
  If $E_1\subset E_2$ and $\ell_1$ and $\ell_2$ are projections,
  then $\Obs_{A_1\leftarrow A_2}(G)$ can be computed in exponential time, and then the
  order-$2$ estimation-based property can be verified in exponential time. 
\end{remark}

\begin{example}\label{exam9:High-OrderCSO}
  Reconsider FSA $G $ studied in \autoref{exam1:High-OrderCSO} (in \autoref{fig1:High-OrderCSO}). We consider agents $A_1'$ and $A_2'$
  whose observable event sets are $E_1'=\{a\}$ and $E_2=\{a,b\}$, respectively,
  and whose labeling functions are projections.
  Then we have $E_1'\subset E_2$. The observer $\Obs_{A_1'}$, the concurrent 
  composition $\CC(G_{A_1'},\Obs_{A_1'})$,
  the order-$2$ observer $\Obs_{A_1'\leftarrow A_2'}(G )$, and the observer
  $\Obs_{A_2}$ are shown in 
  \autoref{fig30:High-OrderOpacity}, \autoref{fig31:High-OrderOpacity}, \autoref{fig32:High-OrderOpacity}, \autoref{fig33:High-OrderOpacity}
  respectively. \autoref{fig32:High-OrderOpacity} and \autoref{fig33:High-OrderOpacity} illustrate \autoref{lem4:High-OrderOpacity}.
  \begin{figure}[!htbp]
  \centering
  \subcaptionbox{The observer $\Obs_{A_1'}$ of automaton $G^{ }_{A_1'}$.\label{fig30:High-OrderOpacity}}{
  	  \begin{tikzpicture}[>=stealth',shorten >=1pt,auto,node distance=2.5 cm, scale = 1.0, transform shape,
	>=stealth,inner sep=2pt]

	\tikzstyle{emptynode}=[inner sep=0,outer sep=0]

	\node[initial, initial where = left, rectangular state] (0) {$\{0\}$};
	\node[rectangular state] (123) [right of =0] {$\{1,2,3\}$};
	\node[rectangular state] (45) [right of =123] {$\{4,5\}$};

	\node[emptynode] (E1) [left of = 0] {};
	\node[emptynode] (E2) [right of = 45] {};

	\path [->]
	(0) edge node {$a$} (123)
	(123) edge node {$a$} (45)
	;

    \end{tikzpicture}
  }\vspace{0.5cm}

  \subcaptionbox{The concurrent composition $\CC(G_{A_1'},\Obs_{A_1'})$.\label{fig31:High-OrderOpacity}}{
	\begin{tikzpicture}[>=stealth',shorten >=1pt,auto,node distance=3.0 cm, scale = 1.0, transform shape,
	>=stealth,inner sep=2pt]

	\node[initial, initial where = left, rectangular state] (0A) {$(0,\{0\})$};
    \node[rectangular state] (1A) [right of =0A] {$(1,\{1,2,3\})$};
	\node[rectangular state] (2B) [right of =1A] {$(2,\{1,2,3\})$};
	\node[rectangular state] (3C) [right of =2B] {$(3,\{1,2,3\})$};
	\node[rectangular state] (4C) [above right of =3C] {$(4,\{4,5\})$};
	\node[rectangular state] (5C) [below right of =3C] {$(5,\{4,5\})$};

	\path [->]
	(0A) edge node {$(a,a)$} (1A)
	(1A) edge node {$(b,\ep)$} (2B)
	(2B) edge [bend left] node {$(b,\ep)$} (3C)
	(3C) edge [bend left] node {$(c,\ep)$} (2B)
	(3C) edge node {$(a,a)$} (4C)
	(3C) edge node {$(a,a)$} (5C)
	(4C) edge [loop right] node {$(d,\ep)$} (4C)
	;

    \end{tikzpicture}
  }\vspace{0.5cm}

  \subcaptionbox{The order-$2$ observer $\Obs_{A_1'\leftarrow A_2'}(G )$.\label{fig32:High-OrderOpacity}}{
	\begin{tikzpicture}[>=stealth',shorten >=1pt,auto,node distance=3.8 cm, scale = 1.0, transform shape,
	>=stealth,inner sep=2pt]

	\node[initial, initial where = left, rectangular state] (0A) {$\{(0,\{0\})\}$};
	\node[rectangular state] (1A) [below =1cm of 0A] {$\{(1,\{1,2,3\})\}$};
	\node[rectangular state] (2B) [right of =1A] {$\{(2,\{1,2,3\})\}$};
	\node[rectangular state] (2B3C) [right of =2B] {$\begin{matrix}\{(2,\{1,2,3\}),\\(3,\{1,2,3\})\}\end{matrix}$};
	\node[rectangular state] (4C5C4D) [right of =2B3C] {$\begin{matrix}\{(4,\{4,5\}),\\(5,\{4,5\})\}\end{matrix}$};

	\path [->]
	(0A) edge node {$a$} (1A)
	(1A) edge node {$b$} (2B)
	(2B) edge node {$b$} (2B3C)
	(2B3C) edge node {$a$} (4C5C4D)
	(2B3C) edge [loop above] node {$b$} (2B3C)
	;

    \end{tikzpicture}

  }\vspace{0.5cm}

  \subcaptionbox{The observer $\Obs_{A_2'}$ of automaton $G_{A_2'}$.\label{fig33:High-OrderOpacity}}{
	\begin{tikzpicture}[>=stealth',shorten >=1pt,auto,node distance=2.5 cm, scale = 1.0, transform shape,
	>=stealth,inner sep=2pt]

	\node[initial, initial where = left, rectangular state] (0A) {$\{0\}$};
	\node[rectangular state] (1A) [right of =0A] {$\{1\}$};
	\node[rectangular state] (2B) [right of =1A] {$\{2\}$};
	\node[rectangular state] (2B3C) [right of =2B] {$\{2,3\}$};
	\node[rectangular state] (4C5C4D) [right of =2B3C] {$\{4,5\}$};

	\path [->]
	(0A) edge node {$a$} (1A)
	(1A) edge node {$b$} (2B)
	(2B) edge node {$b$} (2B3C)
	(2B3C) edge node {$a$} (4C5C4D)
	(2B3C) edge [loop above] node {$b$} (2B3C)
	;

    \end{tikzpicture}
  }\vspace{0.5cm}

  \caption{Illustrative automata in \autoref{exam9:High-OrderCSO}.}
	\end{figure}
\end{example}

\subsubsection{Special case 2 --- \texorpdfstring{$E_2\subset E_1$ and $\ell_1$ and $\ell_2$ are projections}{E2subsetE1}}

Consider $E_2\subset E_1$. In this case, although agent $A_2$ knows less on $G$ than agent
$A_1$, all inferences of $A_1$'s current-state estimate done by $A_2$ are contained
in $A_2$'s own current-state estimate. Formally, for all $\alpha\in \ell_2(L(G))$,
for all $s\in \ell_2^{-1}(\alpha)\cap L(G)$, one has
\[\Mt_1(G,\ell_1(s)) \subset \Mt_2(G,\alpha),\]
because $\ell_2(\ell_1^{-1}(\ell_1(s))) = \{\alpha\}$. This implies
the following result. 

\begin{fact}\label{lem3:High-OrderOpacity}
  Assume $E_2\subset E_1$ and $\ell_1$ and $\ell_2$ are projections. Then $\Obs_{A_1\leftarrow A_2}(G)$
  and $\Obs_{A_2}$ have the same number of reachable states.
  Consider a run $Y_0\xrightarrow[]{\alpha}Y_1$ in observer $\Obs_{A_1\leftarrow A_2}(G)$ and
  a run $X_0\xrightarrow[]{\alpha}X_1$ in observer $\Obs_{A_2}$, where 
  $Y_0$ and $X_0$ are the corresponding initial states, $\alpha\in (E_2)^*$.
  Then for every element $(q,X)$ of $Y_1$, $X\subset X_1$.
\end{fact}

\begin{remark}
  By \autoref{lem3:High-OrderOpacity}, order-$2$ observer $\Obs_{A_1\leftarrow A_2}(G)$ can be computed
  in exponential time, and then the order-$2$ estimation-based property
  can also be verified in exponential time.
\end{remark}

\begin{example}\label{exam10:High-OrderCSO}
  Reconsider FSA $G $ studied in \autoref{exam1:High-OrderCSO} (in \autoref{fig1:High-OrderCSO}). We consider agents $A_1''$ and $A_2'$
  whose observable event sets are $E_1''=\{a,b,d\}$ and $E_2=\{a,b\}$ and 
  whose labeling functions are projections, respectively. Then we have $E_2\subset E_1''$.
  The observer $\Obs_{A_1''}$, the concurrent 
  composition $\CC(G_{A_1''},\Obs_{A_1''})$,
  the order-$2$ observer $\Obs_{A_1''\leftarrow A_2'}(G )$, and the observer
  $\Obs_{A_2'}$ are shown in 
  \autoref{fig34:High-OrderOpacity}, \autoref{fig35:High-OrderOpacity}, \autoref{fig36:High-OrderOpacity}, \autoref{fig33:High-OrderOpacity},
  respectively. \autoref{fig36:High-OrderOpacity} and \autoref{fig33:High-OrderOpacity} illustrate \autoref{lem3:High-OrderOpacity}.
  \begin{figure}[!htbp]
  \centering
  \subcaptionbox{The observer $\Obs_{A_1''}$ of automaton $G_{A_1''}$.\label{fig34:High-OrderOpacity}}{
  	  \begin{tikzpicture}[>=stealth',shorten >=1pt,auto,node distance=2.8 cm, scale = 1.0, transform shape,
	>=stealth,inner sep=2pt]

	\node[initial, initial where = left, rectangular state] (0A) {$\{0\}$};
	\node[rectangular state] (1A) [right of =0A] {$\{1\}$};
	\node[rectangular state] (2B) [right of =1A] {$\{2\}$};
	\node[rectangular state] (2B3C) [right of =2B] {$\{2,3\}$};
	\node[rectangular state] (4C5C4D) [right of =2B3C] {$\{4,5\}$};
	\node[rectangular state] (4) [right of =4C5C4D] {$\{4\}$};

	\path [->]
	(0A) edge node {$a$} (1A)
	(1A) edge node {$b$} (2B)
	(2B) edge node {$b$} (2B3C)
	(2B3C) edge node {$a$} (4C5C4D)
	(2B3C) edge [loop above] node {$b$} (2B3C)
	(4C5C4D) edge node {$d$} (4)
	(4) edge [loop right] node {$d$} (4)
	;

    \end{tikzpicture}
  }\vspace{0.5cm}

  \subcaptionbox{The concurrent composition $\CC(G_{A_1''},\Obs_{A_1''})$.\label{fig35:High-OrderOpacity}}{
	\begin{tikzpicture}[>=stealth',shorten >=1pt,auto,node distance=3.0 cm, scale = 1.0, transform shape,
	>=stealth,inner sep=2pt]

	\tikzstyle{emptynode}=[inner sep=0,outer sep=0]

	\node[initial, initial where = left, rectangular state] (0A) {$(0,\{0\})$};
	\node[rectangular state] (1A) [below = 2.0cm of 0A] {$(1,\{1\})$};
	\node[rectangular state] (2B) [right of =1A] {$(2,\{2\})$};
	\node[rectangular state] (3C) [right of =2B] {$(3,\{2,3\})$};
	\node[rectangular state] (2C) [above =2.0cm of 3C] {$(2,\{2,3\})$};
	\node[rectangular state] (4C) [right of =2C] {$(4,\{4,5\})$};
	\node[emptynode] (empty1) [below =1cm of 3C] {};
	\node[rectangular state] (5C) [right of =empty1] {$(5,\{4,5\})$};
	\node[rectangular state] (4D) [right of =4C] {$(4,\{4\})$};

	\path [->]
	(0A) edge node {$(a,a)$} (1A)
	(1A) edge node {$(b,b)$} (2B)
	(2B) edge node {$(b,b)$} (3C)
	(3C) edge [bend left] node [sloped, above] {$(c,\ep)$} (2C)
	(2C) edge [bend left] node [sloped, above] {$(b,b)$} (3C)
	(3C) edge node [sloped, above] {$(a,a)$} (4C)
	(3C) edge node {$(a,a)$} (5C)
	(4C) edge node {$(d,d)$} (4D)
	(4D) edge [loop above] node {$(d,d)$} (4D)
	;

    \end{tikzpicture}
  }\vspace{0.5cm}

  \subcaptionbox{The order-$2$ observer $\Obs_{A_1''\leftarrow A_2'}(G)$.\label{fig36:High-OrderOpacity}}{
	\begin{tikzpicture}[>=stealth',shorten >=1pt,auto,node distance=3.0 cm, scale = 1.0, transform shape,
	>=stealth,inner sep=2pt]

	\node[initial, initial where = above, rectangular state] (0A) {$\{(0,\{0\})\}$};
	\node[rectangular state] (1A) [right of =0A] {$\{(1,\{1\})\}$};
	\node[rectangular state] (2B) [right of =1A] {$\{(2,\{2\})\}$};
	\node[rectangular state] (2B3C) [right of =2B] {$\begin{matrix}\{(2,\{2,3\}),\\(3,\{2,3\})\}\end{matrix}$};
	\node[rectangular state] (4C5C4D) [right of =2B3C] {$\begin{matrix}\{(4,\{4,5\}),\\(5,\{4,5\}),\\(4,\{4\})\}\end{matrix}$};

	\path [->]
	(0A) edge node {$a$} (1A)
	(1A) edge node {$b$} (2B)
	(2B) edge node {$b$} (2B3C)
	(2B3C) edge node {$a$} (4C5C4D)
	(2B3C) edge [loop above] node {$b$} (2B3C)
	;

    \end{tikzpicture}
  }\vspace{0.5cm}

  \caption{Illustrative automata in \autoref{exam10:High-OrderCSO}.}
	\end{figure}

\end{example}

\subsubsection{Special case 3 --- induced by a \texorpdfstring{$2$}{2}-bounded order-\texorpdfstring{$1$}{1} predicate}
\label{subsubsec:Order2Opacity} 

Consider $2$-bounded order-$1$ predicate
\begin{align}\label{eqn1:spec}
  \mathcal{B}=\{X_1,\dots,X_m\},
\end{align}
where $X_i\subset Q$, $1\le |X_i| \le 2$, $i=1,\dots,m$.
By $\mathcal{B}$, we define a special type of predicates:
\begin{subequations}\label{eqn20:High-OrderOpacity}
\begin{align}
  &\pred^2_{\mathcal{B}} \\
  := & \{\emptyset\notin Y\subset 2^Q | (\exists X\in Y)(\exists X'\in \mathcal{B})[X'\subset X]\}
  \subset 2^{2^Q}.
\end{align}
\end{subequations}

  In $\pred^2_{\mathcal{B}}$, every element (a subset of $2^Q$) has some element $X$ (a subset of $Q$) that contains some $X'$ (an element of $\mathcal{B}$) with cardinality between $1$ and $2$ as its subset.

In order to give an exponential-time verification algorithm, for this special class of predicates, one possibility is to change
the observer $\Obs_{A_1}$ in \autoref{fig1:order2SEproperty} to detector $\Det_{A_1}$.
\begin{enumerate}
  \item Compute the detector $\Det_{A_1}$ of $G_{A_1}$.
  \item Compute the concurrent composition $\CC(G_{A_1},\Det_{A_1})$ of LFSA $G_{A_1}$ and its detector $\Det_{A_1}$,
	replace each event $(e_1,e_2)$ by $e_1$, replace the labeling function 
	by $\ell_2$, and denote the modification of $\CC(G_{A_1},\Det_{A_1})$ by $\CC_{A_1\to A_2}^{G,\Det}$ which is an LFSA.
  \item Compute the observer $\Obs(\CC_{A_1\to A_2}^{G,\Det})$ of $\CC_{A_1\to A_2}^{G,\Det}$.
\end{enumerate}

See \autoref{fig2:order2SEproperty} for an illustration.

\begin{figure}[!htbp]
  \centering
  \begin{tikzpicture}[circuit logic IEC]
	\matrix[column sep=5mm, row sep=0.5cm]
	{
	  \node (SusrR) {$G_{A_1}$}; & & \node (DetusrR) {$\Det_{A_1}$};\\
	  & \node (CCusrR) {$\CC$}; & \\
	  & \node (P1to2R) {$\to \ell_2$}; & \\
	  & \node (ObsIntrR) {$\Obs$}; & \\
	};
	\draw (SusrR.south) -- ++(down:2.0mm) -| (CCusrR.north);
	\draw (DetusrR.south) -- ++(down:2mm) -| (CCusrR.north);
	\draw (CCusrR.south) -- (P1to2R);
	\draw (P1to2R) -- (ObsIntrR.north);
  \end{tikzpicture}
  \caption{Sketch of the $\EXPTIME$ verification structure for the special type of order-$2$ estimation-based
  property with respect to \eqref{eqn20:High-OrderOpacity}.}
  \label{fig2:order2SEproperty}
\end{figure}

The observer $\Obs(\CC_{A_1\to A_2}^{G,\Det})$ can be computed in exponential time.
Note a fundamental difference between $\CC_{A_1\to A_2}^{G,\Det}$ and
  $\CC_{A_1\to A_2}^{G,\Obs}$: in 
  $\CC_{A_1\to A_2}^{G,\Obs}$, for every reachable state $(q,X)$, if there is an observable transition
  starting at $q$ with event $a$ in $G_{A_1}$, then there must exist an observable transition
starting at $(q,X)$ with event $a$ in $\CC_{A_1\to A_2}^{G,\Obs}$. However, this 
may not hold for $\CC_{A_1\to A_2}^{G,\Det}$, see \autoref{fig7:High-OrderCSO}.

\begin{lemma}\label{thm8:High-OrderCSO}
  Consider an FSA $G$ as in \eqref{FSA}, agents $A_1$ and $A_2$ with observable event sets $E_1\subset E$
  and $E_2\subset E$ and labeling functions $\ell_{E_1}$ and $\ell_{E_2}$,
  and the observer $\Obs(\CC_{A_1\to A_2}^{G,\Det})$.
  Then each initial state of observer $\Obs(\CC_{A_1\to A_2}^{G,\Det})$ is of the
	  form $\{(q_{0,1},X_0),\dots,(q_{0,m},X_0)\}$, where 
	  $\{q_{0,1},\dots,q_{0,m}\}=X_0$, $X_0$ is the initial state of $\Det_{A_1}$.
\end{lemma}

\begin{theorem}\label{thm2:order2SEproperty}
  An FSA $G$ satisfies the order-$2$ estimation-based
  property $\pred^2_{\mathcal{B}}$ \eqref{eqn20:High-OrderOpacity} with respect to agents $A_1$ and $A_2$, if and only if,
  for every reachable state  $\{(q_1,X_1),\dots,(q_n,X_n)\}$ of
  $\Obs(\CC_{A_1\to A_2}^{G,\Det})$, $\{X_1,\dots,X_n\}\in \pred^2_{\mathcal{B}}$.
\end{theorem}

\begin{proof}
  It immediately follows from \autoref{lem1:High-OrderOpacity}, \autoref{thm1:order2SEproperty}, and \autoref{thm8:High-OrderCSO}.
\end{proof}

\autoref{thm2:order2SEproperty} provides an exponential-time algorithm for verifying this special type of order-$2$
estimation-based property.

\begin{example}\label{exam2:High-OrderCSO}
  Reconsider FSA $G $ studied in \autoref{exam1:High-OrderCSO} (in \autoref{fig1:High-OrderCSO}). 
  We use \autoref{thm2:order2SEproperty} and follow the sketch shown in 
  \autoref{fig2:order2SEproperty} to verify if $G $ satisfies the order-$2$ predicate
  $\pred^2=\{\emptyset\not\in Y\subset 2^{\{0,1,2,3,4,5\}}| (\exists X\in Y)[|X|>1]\}\subset 2^{2^{\{0,1,2,3,4,5\}}}$
  with respect to agents $A_1^{ }$ and $A_2^{ }$,
  where the corresponding $\mathcal{B}$ as in \eqref{eqn1:spec} is equal to 
  $\{X\subset Q||X|=2\}$. 
  The detector $\Det_{A_1^{ }}$ of $G^{ }_{A_1^{ }}$ is shown in \autoref{fig6:High-OrderCSO}.
  The concurrent composition $\CC(G_{A_1},\Det_{A_1})$ 
  is shown in \autoref{fig7:High-OrderCSO}. The observer $\Obs(\CC_{A_1^{ }\to A_2^{ }}^{G ,\Det})$
  is shown in \autoref{fig8:High-OrderCSO}.

\begin{figure}[!htbp]
  \centering
  \subcaptionbox{The detector $\Det_{A_1^{ }}$ of automaton $G^{ }_{A_1^{ }}$ (shown in \autoref{fig1:High-OrderCSO}).\label{fig6:High-OrderCSO}}{
  	\begin{tikzpicture}[>=stealth',shorten >=1pt,auto,node distance=2.5 cm, scale = 1.0, transform shape,
	>=stealth,inner sep=2pt]

	\tikzstyle{emptynode}=[inner sep=0,outer sep=0]

	\node[initial, initial where = left, rectangular state] (01) {$\{0,1\}$};
	\node[rectangular state] (2) [right of =01] {$\{2\}$};
	\node[rectangular state] (45) [right of =2] {$\{4,5\}$};
	\node[rectangular state] (34) [above = 1cm of 45] {$\{3,4\}$};
	\node[rectangular state] (35) [below = 1cm of 45] {$\{3,5\}$};
	\node[rectangular state] (4) [right of =34] {$\{4\}$};

	\node[emptynode] (E1) [left of = 01] {};
	\node[emptynode] (E2) [right of = 4] {};
	
	\path [->]
	(01) edge node {$\alpha$} (2)
	(2) edge node {$\alpha$} (34)
	(34) edge [bend right] node [above] {$\alpha$} (2)
	(34) edge node {$\beta$} (4)
	(4) edge [loop right] node {$\beta$} (4)

	(2) edge node {$\alpha$} (45)
	(45) edge node {$\beta$} (4)

	(2) edge node {$\alpha$} (35)
	(35) edge [bend left] node [below] {$\alpha$} (2)
	;
    \end{tikzpicture}
  }

  \subcaptionbox{The concurrent composition $\CC(G_{A_1},\Det_{A_1})$,
	where $A=\{0,1\}$, $B=\{2\}$, $C_1=\{3,4\}$, $C_2=\{4,5\}$, $C_3=\{3,5\}$, $D=\{4\}$. 
  Note that at state $(4,C_3)$ there is no transition with event $(d,\beta)$, although there is an observable transition $4\xrightarrow[]{d}4$ in $G_{A_1^{ }}^{ }$.\label{fig7:High-OrderCSO}}{
	\begin{tikzpicture}[>=stealth',shorten >=1pt,auto,node distance=3.1 cm, scale = 1.0, transform shape,
	>=stealth,inner sep=2pt]

	\node[initial, initial where = left, rectangular state] (0A) {$(0,A)$};
	\node[rectangular state] (1A) [below =1cm of 0A] {$(1,A)$};
	\node[rectangular state] (2B) [right of =1A] {$(2,B)$};
	\node[rectangular state] (3C2) [right of =2B] {$(3,C_2)$};
	\node[rectangular state] (3C1) [above of =3C2] {$(3,C_1)$};
	\node[rectangular state] (3C3) [below of =3C2] {$(3,C_3)$};
	\node[rectangular state] (5C1) [right of =3C1] {$(5,C_1)$};
	\node[rectangular state] (4C3) [right of =3C3] {$(4,C_3)$};
	\node[rectangular state] (4C2) [below =0.3cm of 5C1] {$(4,C_2)$};
	\node[rectangular state] (4C1) [above =0.3cm of 5C1] {$(4,C_1)$};
	\node[rectangular state] (5C2) [above =0.3cm of 4C3] {$(5,C_2)$};
	\node[rectangular state] (5C3) [below =0.3cm of 4C3] {$(5,C_3)$};
	\node[rectangular state] (4D) [right of =4C1] {$(4,D)$};
	
	\path [->]
	(0A) edge node {$(a,\ep)$} (1A)
	(1A) edge node {$(b,\alpha)$} (2B)
	(2B) edge node {$(b,\alpha)$} (3C2)
	(2B) edge node [sloped, above] {$(b,\alpha)$} (3C1)
	(2B) edge node [sloped, above] {$(b,\alpha)$} (3C3)

	(3C1) edge [bend right] node [sloped, above] {$(c,\alpha)$} (2B)
	(3C3) edge [bend left] node [sloped, above] {$(c,\alpha)$} (2B)
	(3C1) edge node [sloped, above] {$(a,\ep)$} (4C1)
	(3C1) edge node {$(a,\ep)$} (5C1)

	(3C2) edge node [sloped, above] {$(a,\ep)$} (4C2)
	(3C2) edge node [sloped, above] [below] {$(a,\ep)$} (5C2)

	(3C3) edge node {$(a,\ep)$} (4C3)
	(3C3) edge node [sloped, above] {$(a,\ep)$} (5C3)

	(4C1) edge node {$(d,\beta)$} (4D)
	(4C2) edge node [sloped, above] {$(d,\beta)$} (4D)

	(4D) edge [loop right] node {$(d,\beta)$} (4D)
	;

    \end{tikzpicture}
  }\vspace{0.5cm}

  \subcaptionbox{The observer $\Obs(\CC_{A_1^{ }\to A_2^{ }}^{G ,\Det})$.\label{fig8:High-OrderCSO}}{
  \begin{tikzpicture}[>=stealth',shorten >=1pt,auto,node distance=2.5 cm, scale = 1.0, transform shape,
	>=stealth,inner sep=2pt]

	\node[initial, initial where = left, rectangular state] (0A) {$\{(0,A)\}$};
	\node[rectangular state] (1A) [right of =0A] {$\{(1,A)\}$};
	\node[rectangular state] (2B) [right of =1A] {$\{(2,B)\}$};
	\node[rectangular state] (2B3C1-3) [right of =2B] {$\begin{matrix}\{(2,B),\\(3,C_1),\\(3,C_2),\\(3,C_3)\}\end{matrix}$};
	\node[rectangular state] (4D4C1-35C1-3) [right of =2B3C1-3] {$\begin{matrix}\{(4,D),\\(4,C_1),\\(5,C_1),\\(4,C_2),\\(5,C_2),\\(4,C_3),\\(5,C_3)\}\end{matrix}$};

	\path [->]
	(0A) edge node {$\gamma$} (1A)
	(1A) edge node {$\delta$} (2B)
	(2B) edge node {$\delta$} (2B3C1-3)
	(2B3C1-3) edge node {$\gamma$} (4D4C1-35C1-3)
	(2B3C1-3) edge [loop above] node {$\delta$} (2B3C1-3)
	;

    \end{tikzpicture}
  }\vspace{0.5cm}

 \caption{Illustrative automata in \autoref{exam2:High-OrderCSO}.}
  \end{figure}

  In \autoref{fig8:High-OrderCSO},
  there is a reachable state $\{(2,B)\}$ in which 
  $B$ is a singleton, and there is no state of the form
  $(q,X)$ with $|X|=2$. By \autoref{thm2:order2SEproperty}, $G $ does not satisfy the order-$2$
  estimation-based property, which is consistent with the result derived in \autoref{exam1:High-OrderCSO}.
\end{example}

\begin{remark} 
Based on the above argument, whether an FSA $G$ satisfies the order-$2$
estimation-based
property $2^{2^Q}\setminus\pred^2_{\mathcal{B}}$  with respect to agents $A_1$ and $A_2$
can also be verified in exponential time, where $\pred^2_{\mathcal{B}}$ is defined in
\eqref{eqn20:High-OrderOpacity}.
\end{remark}

\subsection{Order-\texorpdfstring{$2$}{2} current-state opacity}
\label{subsec:Order2Opacity}

In this subsection, we study a special type of order-$2$ estimation-based property --- order-$2$
current-state opacity.

\begin{figure}[!htbp]
  \centering
  \subcaptionbox{Order-$1$ opacity. Based on the observation $\alpha$ of $\Intr$ to $G$, $\Intr$ infers the set $Q'$ of all consistent states of $G$.\label{IllustrOrd1Opacity}}{
  	  \begin{tikzpicture}[>=stealth',shorten >=1pt,auto,node distance=6.0 cm, scale = 1.0, transform shape,
		  >=stealth,inner sep=5pt,
		  every text node part/.style={align=center}
	  ]

	\node[rectangular state] (plant) {$G$\\ (fictitious agent $\Usr$)};
	\node[rectangular state, right = 7cm of plant] (intruder) {intruder $\Intr$};

	\path [->]
	;

	\draw [->]
	($(intruder.west)+(0,0.1cm)$) -- ++(0,0) node [above, sloped, xshift=-96] {infer $G|\alpha=Q'$} -- ($(plant.east)+(0,0.1cm)$)
	;
	\draw [->]
	($(plant.east)+(0,-0.1cm)$) -- ++(0,0) node [below, sloped, xshift=100] {observation $\alpha$} -- ($(intruder.west)+(0,-0.1cm)$)
	;

    \end{tikzpicture}
  }\vspace{0.5cm}

  \subcaptionbox{Order-$2$ opacity. Given an event sequence $s\in L(G)$, based on the observation $\alpha=ell_{\Usr}(s)$ of $\Usr$ to $G$, $\Usr$ infers the set $Q'$ of all consistent states of $G$, and then based on the observation $\beta=\ell_{\Intr}(s)$ of $\Intr$ to $G$, $\Intr$ infers $Q'$ and obtains a subset of $2^Q$.\label{IllustrOrd2Opacity}}{
	\begin{tikzpicture}[>=stealth',shorten >=1pt,auto,node distance=6.0 cm, scale = 1.0, transform shape,
		  >=stealth,inner sep=5pt,
	every text node part/.style={align=center}
	  ]

	\node[rectangular state] (plant1) {$G$};
	\node[rectangular state, right = 4cm of plant1] (Usr1) {User $\Usr$};
	\node[rectangular state, right = 4cm of Usr1] (Intr1) {Intruder $\Intr$};

	\path [->]
	(Intr1) edge node [above, sloped] {infer $Q'|\beta\subset 2^Q$} (Usr1)
	;

	\draw [->] ($(plant1.east)+(0,-0.1cm)$) -- ++(0,0) node [below, sloped, xshift=50, yshift=2] {observation $\alpha$} -- ($(Usr1.west)+(0,-0.1cm)$);

	\draw [->] ($(Usr1.west)+(0,0.1cm)$) -- ++(0,0) node [above, sloped, xshift=-60, yshift=-2] {infer $G|\alpha=Q'$} -- ($(plant1.east)+(0,0.1cm)$);

	\draw [->] (plant1.south) -- ++(0,-0.5) node [below, xshift=130] {observation $\beta$} -| (Intr1.south)
	;

    \end{tikzpicture}
  }

	\caption{An illustration of the extension from order-$1$ opacity to order-$2$ opacity.}
	\label{fig49:High-OrderCSO} 
  \end{figure}

Let us first review current-state opacity studied in \autoref{subsec:CSO}
in a high level. A fictitious ``user'' (which may assume coincide with $G$) $\Usr$ (\romannumeral1) knows the structure
of an automaton $G$, also (\romannumeral2) knows the state $G$ is in at every instant, and 
(\romannumeral3) wants 
to determine if the fact that $G$ is currently in a secret state can be disclosed
to an ``intruder'' $\Intr$
who also knows
the structure of $G$ but can only see the occurrences of observable events of
$E_{\Intr}$.
See \autoref{IllustrOrd1Opacity} as an illustration.
If $G$ is sufficiently safe in that sense, that is, the fact is not disclosed
to $\Intr$, then $\Usr$ can 
operate on $G$. In other words, $\Usr$ knows that $\Intr$ cannot not detect
if $\Usr$ knows the fact, and then $\Usr$ will operate on $G$.
This can be regarded as order-$1$ estimation-based property. 
However, if $\Usr$ knows some knowledge
of $G$ (although less than before) but still can do (\romannumeral3), then $G$ can
also be regarded to be sufficiently safe. This can be formulated as \emph{order-$2$
state-based opacity}. See \autoref{IllustrOrd2Opacity} as an illustration.
In this more general case, we assume the user still satisfies
(\romannumeral1) but does not satisfy (\romannumeral2). Instead, we assume that
$\Usr$ can observe a subset $E_{\Usr}$ of events of $G$, so can do state estimation
according
to its observations to $G$. We also assume that $\Intr$ knows $E_{\Usr}$ although
cannot observe $E_{\Usr}\setminus E_{\Intr}$, so can infer what $\Usr$ observes 
according to $\Intr$'s own observations. It turns out that $\Intr$'s inference 
of $\Usr$'s state estimate of $G$ is a set of subsets of $G$. Hence in this more 
general case, the secrets corresponding to $\Intr$ are a set of subsets of states of $G$
instead of a subset $\QS$ of states of $G$ as in the order-$1$ case. We define 
\emph{order-$2$ secrets} as 
\begin{align}\label{eqn2':High-OrderOpacity}
  \QSordtwo\subset 2^Q.
\end{align}

A predicate is defined as
\begin{subequations}\label{eqn3':High-OrderOpacity}
\begin{align}
  &\pred^2(G,\QSordtwo)\\
  := & \{\emptyset\notin Y\subset 2^Q|Y\not\subset\QSordtwo\}\subset 2^{2^Q}.
\end{align}
\end{subequations}

For $\alpha\in \ell_{\Intr}(L(G))$, $\Intr$'s inference of $\Usr$'s current-state estimate of $G$
can be any $\Mt_{\Usr}(G,\ell_{\Usr}(s))$, where $s \in \ell_{\Intr}^{-1}(\alpha)\cap L(G)$, formulated as 
\begin{subequations}\label{eqn8:High-OrderOpacity}
\begin{align}
  & \Mt_{\Usr\leftarrow\Intr}(G,\alpha) \label{eqn8_1:High-OrderOpacity}\\
  := & \{\Mt_{\Usr}(G,\ell_{\Usr}(s))|s\in \ell_{\Intr}^{-1}(\alpha)\cap L(G)\}
  \subset 2^Q.\label{eqn8_2:High-OrderOpacity}
\end{align}
\end{subequations}
Note that \eqref{eqn8:High-OrderOpacity} actually coincides with \eqref{eqn15:High-OrderOpacity} if agents $A_1$
and $A_2$ are specified as the user $\Usr$ and the intruder $\Intr$, respectively.

Because $\Intr$ computes all possible strings $s\in L(G)$ with $\ell_{\Intr}(s)=\alpha$
which must contain the real generated string,
$\Mt_{\Usr\leftarrow\Intr}(G,\alpha)$ must contain the real current-state estimate
of $\Usr$.

\begin{definition}[Ord$2$CSO]\label{def5':High-OrderCSO} 
  An FSA $G$ is called \emph{order-$2$ current-state opaque with respect to $\Usr$, $\Intr$, and $\QSordtwo$} if
\begin{subequations}\label{eqn4:High-OrderOpacity}
  \begin{align}
	&\{\Mt_{\Usr\leftarrow\Intr}(G,\alpha)|\alpha\in \ell_{\Intr}(L(G))\} \\
	\subset & \pred^2(G,\QSordtwo).
  \end{align}
\end{subequations}
\end{definition}


\autoref{def5':High-OrderCSO} means that if FSA $G$ is order-$2$ current-state
opaque with respect to $\ell_{\Usr}$, $\ell_{\Intr}$, and $\QSordtwo$, then corresponding to every 
observation 
$\alpha\in \ell_{\Intr} (L(G))$ of intruder $\Intr$ to $G$, the inference $
\Mt_{\Usr\leftarrow\Intr}(G,\alpha)$ of the current-state estimate of
user $\Usr$ by $\Intr$ belongs to the predicate $\pred^2(G,\QSordtwo)$.

\begin{remark}\label{rem1:High-OrderOpacity}
  By definition, one sees that the order-$2$ current-state opacity of FSA $G$ with respect to $\Usr$, $\Intr$, and
  $\QSordtwo$ is a special case of the order-$2$ estimation-based property
  $\pred^2$ of $G$ with respect to agents $A_1$ and $A_2$ as in \autoref{def:order2SEproperty} because
  $\Usr$ and $\Intr$ can be regarded as agents $A_1$ and $A_2$, respectively. For this special case, the order-$2$
  estimation-based property is used to describe a practical scenario.
\end{remark}

In the Appendix \ref{appendix}, we compare our approach with \cite{Cui2022YouDontKnowWhatIKnow} in which
a special subclass of order-$2$ current-state opacity was investigated.


\section{Order-\texorpdfstring{$n$}{n} estimation-based problems}
\label{sec:Order3Property}

In this section, we formulate order-$n$ estimation-based problems.
Recall that an order-$1$ (estimation-based) property characterizes a single agent $A$'s inference of
the current state of an FSA $G$ based on its observation to $G$. Also recall that an order-$2$ property
involves in two ordered agents $A_1$ and $A_2$, where $A_1$ infers the current state of $G$ based on its
observation to $G$, $A_2$ infers $A_1$'s current-state estimate of $G$ based $A_2$'s own observation to $G$ 
and $A_1$'s observable events (do not forget that $A_2$ knows $A_1$'s observable events but can 
only observe
its own observable events). In this section, we extend this framework to $n$ ordered agents with $n\ge3$.
Particularly, when $n=3$, the third agent $A_3$ infers $A_2$'s inference of $A_1$'s current-state estimate 
of $G$ based on $A_3$'s observation to $G$ and $A_1$ and $A_2$'s observable events.
The general framework is formalized as follows.

Consider an 
FSA $G$ as in \eqref{FSA} and agents $A_i$ with observable event sets $E_i\subset E$ and labeling functions
$\ell_i:E_i\to\Sig_i, E\setminus E_i\to\{\ep\}$, where $\Sig_i$ are alphabets,
$i\in\llb 1,n\rrb$. Assume all agents know the structure of $G$.
Assume $A_i$ knows $E_{k_i}$ but cannot
observe events of $E_{k_i}\setminus E_i$, $i\in\llb 2,n\rrb$, $k_i\in \llb 1,i-1\rrb$.
We characterize what $A_n$ knows
about what $A_{n-1}$ knows about \dots what $A_2$ knows about $A_1$'s state estimate
of $G$.

For state set $Q$, denote $\Pow(Q)=\Pow_1(Q)=2^Q$. It is recursively extended as follows:
for $n\in \Z_+$, $\Pow_{n+1}(Q):=\Pow(\Pow_{n}(Q))$. For example, $\Pow_2(Q)=2^{2^Q}$.

\subsection{The general framework}
Denote
\begin{subequations}\label{eqn38:High-OrderOpacity} 
  \begin{align}
	\ell_{E_i} &=: \ell_i, & \Obs_{\ell_{E_i}}(G) &=: \Obs_i(G),\\
	\Det_{\ell_{E_i}}(G) &=: \Det_i(G), & \Mt_{\ell_{E_i}}(G,\alpha) &=: \Mt_i(G,\alpha),
  \end{align}
\end{subequations}
for short, $i\in\llb1,n\rrb$.

Given a label sequence $\alpha\in \ell_n(L(G))$ observed by agent $A_n$, a consistent 
event sequence for $A_n$ can be any $s_{n}\in \ell_n^{-1}(\alpha)\cap L(G)$, and the label 
sequence observed by agent $A_{n-1}$ can be $\ell_{n-1}(s_{n})$ for any such $s_{n}$; a 
consistent event sequence for $A_{n-1}$ can be any $s_{n-1}\in \ell_{n-2}^{-1}(\ell_{n-1}(s_{n}))\cap L(G)$
for any such $s_{n}$, and the label sequence 
observed by agent $A_{n-2}$ can be $\ell_{n-2}(s_{n-1})$ for any such $s_{n-1}$; \dots; a consistent 
event sequence for $A_2$ can be any $s_2\in \ell_1^{-1}(\ell_2(s_3))\cap L(G)$ for any such $s_3$, and the label 
sequence observed by agent $A_1$ can be $\ell_1(s_2)$ for any such $s_2$, and the current-state
estimate of $A_1$ can be $\Mt_1(G,\ell_1(s_2))$ for any such $s_2$.
See \autoref{fig37:High-OrderOpacity} and \autoref{FA:fig48} as illustrations.
Based on the label sequence
$\alpha\in \ell_n(L(G))$ observed by agent $A_n$,
the order-$n$ current-state estimate of $G$ is formulated as
\begin{subequations}
  \begin{align*}
	&\Mt_{A_1\leftarrow A_2\leftarrow \cdots \leftarrow A_n} (G,\alpha) \\
	:= &
	\overbrace{\{\{\dots\{}^{n-1} \Mt_1(G,\ell_1(s_2)) | 
		s_2 \in \ell_1^{-1}(\ell_2(s_3))\cap L(G) \} | \\
		&\qquad\qquad\qquad\qquad\quad\ \ \dots \\
	    &\qquad\qquad\qquad\qquad\quad\ \ s_{n-1}\in \ell_{n-2}^{-1}(\ell_{n-1}(s_{n}))\\
	    &\qquad\qquad\qquad\qquad\quad\ \ \cap L(G) \} | \\
   	    &\qquad\qquad\qquad\qquad\quad\ \  s_{n}\in \ell_n^{-1}(\alpha)\cap L(G) \}\\
		\subset & \Pow_{n-1}(Q).
  \end{align*}
\end{subequations}

\begin{figure}[!htbp]
  \centering
	\begin{tikzpicture}[>=stealth',shorten >=1pt,auto,node distance=3.5 cm, scale = 1.0, transform shape,
	>=stealth,inner sep=2pt]

	\tikzstyle{emptynode}=[rectangle, draw, inner sep=0,outer sep=0, minimum size = 1cm]
	\tikzstyle{emptystate}=[inner sep=0,outer sep=0]

	\node[emptynode] (alpha) {$\alpha\in \ell_n(L(G))$};
	\node[emptynode] (An) [right of = alpha] {$A_n$};
	\node[emptynode] (psn-1) [above = 1cm of alpha] {$\ell_{n-1}(s_{n})$};
	\node[emptynode] (An-1) [right of = psn-1] {$A_{n-1}$};
	\node[emptynode] (sn-1) [left of = psn-1] {$s_{n}\in \ell_n^{-1}(\alpha)\cap L(G)$};
	\node[emptynode] (s2) [above = 1cm of sn-1] {$s_{n-1}\in \ell_{n-2}^{-1}(\ell_{n-1}(s_{n}))\cap L(G)$};
	\node[emptystate] (s2') [above = 1cm of s2] {$\vdots$};
	\node[emptystate] (empty1) [above = 1cm of psn-1] {$\vdots$};
	\node[emptynode] (ps2) [above = 1cm of empty1] {$\ell_2(s_3)$};
	\node[emptynode] (A2) [right of = ps2] {$A_2$};
	\node[emptystate] (empty2) [above = 1cm of An-1] {$\vdots$};
	\node[emptynode] (s1) [above = 1cm of s2'] {$s_2\in \ell_1^{-1}(\ell_2(s_3))\cap L(G)$};
	\node[emptynode] (ps1) [right of = s1] {$\ell_1(s_2)$};
	\node[emptynode] (A1) [right of = ps1] {$A_1$};
	\node[emptynode] (MA1) [above = 1cm of s1] {$\Mt_1(G,\ell_1(s_2))$};

	\path [->]
	(alpha) edge (sn-1)
	(sn-1) edge (psn-1)
	(psn-1) edge (s2)
	(s1) edge (ps1)
	(ps1) edge (MA1)
	(ps2) edge (s1)
	;

    \end{tikzpicture}
	\caption{Illustration of order-$n$ current-state estimate $\Mt_{A_1\leftarrow A_2\leftarrow \cdots \leftarrow A_n} (G,\alpha)$ of $G$ with respect to $\alpha\in \ell_n(L(G))$.}
    \label{fig37:High-OrderOpacity}
\end{figure}

\begin{figure}[!htbp]
  \centering
  	\begin{tikzpicture}[>=stealth',shorten >=1pt,auto,node distance=7.0 cm, scale = 1.0, transform shape,
		  >=stealth,inner sep=5pt,
	every text node part/.style={align=center}
	  ]

	\tikzstyle{emptynode}=[inner sep=0,outer sep=0]

	\node[rectangular state] (plant) {$G$};
	\node[rectangular state, right = 2cm of plant] (A1) {$A_1$};
	\node[rectangular state, right = 2cm of A1] (A2) {$A_2$};
	\node[emptynode, right = 1cm of A2] (Adot) {$\cdots$};
	\node[rectangular state, right = 3.2cm of Adot] (An) {$A_n$};

	\node[emptynode, below = 1.0cm of A1] (E1) {};
	\node[emptynode, right = 1.3cm of E1] (E2) {$\vdots$};

	\path [->]
	(A2) edge node [above, sloped] {infer $Q^1|\alpha_2$\\$=Q^2\subset 2^Q$} (A1)
	(Adot) edge (A2)
	(An) edge node [above, sloped] {infer $Q^{n-1}|\alpha_n$\\$=Q^{n}\subset \Pow_{n-1}(Q)$} (Adot)
	;

	\draw [->] ($(A1.west)+(0,0.1cm)$) -- ++(0,0) node [above, sloped, xshift=-30] {infer $G|\alpha_1$\\$=Q^1\subset Q$} -- ($(plant.east)+(0,0.1cm)$);

	\draw [->] ($(plant.east)+(0,-0.1cm)$) -- ++(0,0) node [below, sloped, xshift=30] {$\alpha_1$} -- ($(A1.west)+(0,-0.1cm)$);

	\draw [->] (plant.south) -- ++(0,-0.5) node [below, xshift=80] {$\alpha_2$} -| (A2.south)
	;
	\draw [->] (plant.south) -- ++(0,-2.0) node [below, xshift=160] {$\alpha_n$} -| (An.south)
	;

    \end{tikzpicture}
	\caption{An illustration of order-$n$ estimation-based property.}
	\label{FA:fig48}
\end{figure}

As a particular case, \eqref{eqn15:High-OrderOpacity} is the order-$2$ current-state estimate of $G$ 
with respect to $\alpha\in \ell_2(L(G))$. Similarly, the order-$3$ current-state
estimate of $G$ with respect to $\alpha\in \ell_3(L(G))$ is 
\begin{equation}\label{eqn32:High-OrderOpacity}
  \begin{split}
	& \Mt_{A_1\leftarrow A_2\leftarrow A_3} (G,\alpha) \\
	= &
	  \{\{ \Mt_1(G,\ell_1(s_2)) | 
		s_2 \in \ell_1^{-1}(\ell_2(s_3))\cap L(G) \} | \\
		&\qquad\qquad\qquad\quad\ \, s_{3}\in \ell_3^{-1}(\alpha)\cap L(G) \}\\
	\subset & 2^{2^Q}.
  \end{split}
\end{equation}

Define a \emph{predicate of order-$n$ as} 
\begin{align}\label{eqn:orderNpredicate}
  \pred^n \subset \Pow_n(Q).
\end{align}
Then an order-$n$ estimation-based property is defined as follows.
\begin{definition}\label{def:orderNSEproperty}
  An FSA $G$ satisfies the \emph{order-$n$ estimation-based
  property $\pred^n$~\eqref{eqn:orderNpredicate} with respect to agents $A_1,\dots,A_n$} if 
  \begin{equation}\label{eqn33:High-OrderOpacity} 
	\{\Mt_{A_1\leftarrow A_2\leftarrow \cdots \leftarrow A_n}(G,\alpha)|\alpha\in \ell_n(L(G))\}\subset\pred^n.
  \end{equation}
\end{definition}

The next is to define a notion of \emph{order-$n$ observer} to verify \autoref{def:orderNSEproperty}.
Apparently, the classical observer $\Obs_{A_1}$ as in \autoref{FA:def_obs} is the order-$1$ observer,
$\Obs(\CC_{A_1\to A_2}^{G,\Obs})$ as in \autoref{def:order2observer}, which is the observer of 
concurrent composition $\CC(G_{A_1},\Obs_{A_1})$, is the order-$2$ observer. 

The following sequence of concurrent compositions provide a foundation for the order-$n$ observer to be defined.

\begin{enumerate}
  \item[\mylabel{item9:SEproperty}{(1)}]
	Define  $\CC_{A_1\to A_2}^{G,\Obs}$ as before.
  \item[\mylabel{item10:SEproperty}{(2)}]
	Compute concurrent composition $\CC(G_{A_2},\Obs(\CC_{A_1\to A_2}^{G,\Obs}))$, 
	replace each event $(e_1,e_2)$ by $e_1$, replace the labeling function of $\CC(G_{A_2},\Obs(\CC_{A_1\to A_2}^{G,\Obs}))$
	by $\ell_3$, and denote the modification of $\CC(G_{A_2},\Obs(\CC_{A_1\to A_2}^{G,\Obs}))$ by
	$\CC_{A_1\to A_2\to A_3}^{G,\Obs(G,\Obs)}$.
  \item [] $\vdots$
  \item[\mylabel{item11:SEproperty}{(n-1)}]
	Compute concurrent composition $\CC(G_{A_{n-1}},\Obs(\CC_{A_1\to A_2 \to \cdots A_{n-2} \to
	A_{n-1}}^{G,\Obs(G,\Obs(\dots (G,\Obs\overbrace{\scriptstyle)\dots)}^{n-3}}))$, replace each event $(e_1,e_2)$ by $e_1$, replace the 
	labeling function of $\CC(G_{A_{n-1}},\Obs(\CC_{A_1\to A_2 \to \cdots A_{n-2} \to A_{n-1}}^{G,\Obs (G,
	  \Obs(\dots (G,\Obs\overbrace{\scriptstyle)\dots)}^{n-3}}))$ by $\ell_n$, and denote the modification of $\CC(G_{A_{n-1}},\Obs(\CC_{A_1\to
	  A_2 \to \cdots A_{n-2} \to A_{n-1}}^{G,\Obs (G,\Obs(\dots (G,\Obs\overbrace{\scriptstyle)\dots)}^{n-3}}))$ by\\
	  $\CC_{A_1\to A_2 \to \cdots \to A_{n-1}\to A_n}^{G,\Obs(G,\Obs(\dots(G,\Obs\overbrace{\scriptstyle)\dots)}^{n-2}}$.
  \item [] $\vdots$
\end{enumerate}

Then define the order-$n$ observer as follows.

\begin{definition}\label{def:orderNobserver}
   Consider an FSA $G$ as in \eqref{FSA} and agents $A_i$ with observable event sets $E_i\subset E$
   and labeling functions $\ell_i:E_i\to\Sig_i, E\setminus E_i\to\{\ep\}$, where $\Sig_i$ are alphabets,
   $i\in\llb 1,n\rrb$. The \emph{order-$1$ observer} is defined as $\Obs_{A_1}$.
   For each $n>1$, the \emph{order-$n$ observer} $\Obs_{A_1\leftarrow A_2\leftarrow \cdots \leftarrow A_{n}}
   (G)$ is defined as $\Obs(\CC_{A_1\to A_2 \to \cdots \to A_{n-1}\to A_n}^{G,\Obs(G,\Obs(\dots(G,\Obs\overbrace
   {\scriptstyle)\dots)}^{n-2}})$.
\end{definition}

The order-$n$ observer $\Obs_{A_1\leftarrow A_2\leftarrow \cdots \leftarrow A_{n}} (G)$
can be computed in $n$-$\EXPTIME$.

Similar to \autoref{thm7:High-OrderCSO}, the following result holds and can be proven by mathematical induction
easily. 

\begin{theorem}\label{thm9:High-OrderCSO}   
  Consider an FSA $G$ as in \eqref{FSA}, agents $A_i$ with observable event sets $E_i\subset E$
  and labeling functions $\ell_i:E_i\to\Sig_i, E\setminus E_i\to\{\ep\}$, where $\Sig_i$ are alphabets,
  $i\in\llb 1,n\rrb$, and the order-$n$ observer $\Obs_{A_1\leftarrow A_2\leftarrow \cdots \leftarrow A_n} (G)$.
  \begin{enumerate}[(i)] 
	\item\label{item7:SEproperty}
	  $L(G)=L(\CC(G_{A_n},\Obs_{A_1\leftarrow A_2\leftarrow \cdots \leftarrow A_n} (G)))$.
	\item\label{item8:SEproperty}
	  For every $\alpha\in \ell_n(L(G))$ and every run $\mathcal{X}_0\xrightarrow[]{\alpha} \mathcal{X}$
	  of order-$n$ observer $\Obs_{A_1\leftarrow A_2\leftarrow \cdots \leftarrow A_n} (G)$, where
	  $\mathcal{X}_0$ is the initial state, 
	  for $\mathcal X$, replace each ordered pair $(\#,\$)$ by $\$$, where
	  $\#$ is some state of $G$, denote the most updated $\mathcal{X}$ by
	  $\mathcal{\bar X}$, then $\mathcal{\bar X} = \Mt_{A_1\leftarrow A_2\leftarrow \cdots \leftarrow A_n}(G,\alpha)\in \Pow_n(Q)$. 
  \end{enumerate}
\end{theorem}

How to compute $\mathcal{\bar X}$ is illustrated in \autoref{exam11:High-OrderCSO}.

Similar to \autoref{thm1:order2SEproperty}, the following \autoref{thm1:orderNSEproperty} holds.

\begin{theorem}\label{thm1:orderNSEproperty}
  An FSA $G$ satisfies the order-$n$ estimation-based
  property $\pred^n$ \eqref{eqn:orderNpredicate}
  with respect to agents $A_1,\dots,A_n$, if and only if, for every reachable state $\mathcal{X}$
  of the order-$n$ observer $\Obs_{A_1\leftarrow A_2\leftarrow \cdots \leftarrow A_n} (G)$,
  $\mathcal{\bar X}\in \pred^n$, where $\mathcal{\bar X}$ is as in \autoref{thm9:High-OrderCSO}.
\end{theorem}

\begin{remark}
  \autoref{thm1:orderNSEproperty} provides an $n$-$\EXPTIME$ algorithm for verifying \autoref{def:orderNSEproperty}.
\end{remark}

\subsection{Special cases}

Similar to \autoref{subsec:Order2SEproperty}, in this subsection, we discuss special
cases for which the complexity of verifying the order-$n$ estimation-based property
can be significantly reduced. To this end, we define \emph{level-$i$ elements} of $\pred^n$
\eqref{eqn:orderNpredicate}. The \emph{level-$1$ elements} of $\pred^n$ are defined
as elements of $\pred^n$, the \emph{level-$2$ elements} of $\pred^n$ are defined as
elements of level-$1$ elements of $\pred^n$, \dots,  the \emph{level-$(i+1)$ elements}
of $\pred^n$ are defined as elements of level-$i$ elements of $\pred^n$, \dots,
the \emph{level-$n$ elements} of $\pred^n$ are defined as
elements of level-$(n-1)$ elements of $\pred^n$. Then by definition, each level-$i$
element of $\pred^n$ is a subset of $\Pow_{n-i}(Q)$, $i\in\llb 1,n \rrb$.
Particularly, each level-$n$ element of $\pred^n$ is a subset of $Q$.

\subsubsection{Special case 1 --- under containment relations}

Assume for each $i\in\llb 1,n-1\rrb$, either $E_i\subset E_{i+1}$ or $E_{i+1}\subset
E_i$, and $\ell_i$ is a projection. Also assume $\ell_n$ is a projection.

As mentioned before, the order-$1$ observer $\Obs_{A_1}$ can be computed in exponential time.
By \autoref{lem4:High-OrderOpacity} and \autoref{lem3:High-OrderOpacity}, 
the order-$2$ observer $\Obs_{A_1\leftarrow A_2} (G)$ can be computed in time polynomial
in the size of $\Obs_{A_1}$, furthermore, the order-$3$ observer $\Obs_{A_1\leftarrow A_2\leftarrow A_3} (G)$ can be computed in time polynomial
in the size of $\Obs_{A_1\leftarrow A_2} (G)$, \dots,
finally, the order-$n$ observer $\Obs_{A_1\leftarrow A_2\leftarrow \cdots \leftarrow A_n} (G)$
can be computed in time polynomial in the size of the order-$(n-1)$
observer $\Obs_{A_1\leftarrow A_2\leftarrow \cdots \leftarrow A_{n-1}} (G)$. 
As a summary, the order-$n$ observer can be computed in exponential time.

\begin{theorem}\label{thm3:orderNSEproperty} 
  Assume for each $i\in\llb 1,n-1\rrb$, either $E_i\subset E_{i+1}$ or $E_{i+1}\subset E_i$,
  and $\ell_i$ is a projection.  Also assume $\ell_n$ is a projection.
  Then the order-$n$ observer $\Obs_{A_1\leftarrow A_2\leftarrow \cdots \leftarrow A_n} (G)$
  can be computed in exponential time, resulting in whether an FSA $G$ satisfies the order-$n$ estimation-based
  property $\pred^n$~\eqref{eqn:orderNpredicate} with respect to agents $A_1,\dots,A_n$
  can also be verified in exponential time.
\end{theorem}

\subsubsection{Special case 2 --- induced by a \texorpdfstring{$2$}{2}-bounded order-\texorpdfstring{$1$}{1} predicate}

Reconsider $\mathcal{B}$ as in \eqref{eqn1:spec}. By this $\mathcal{B}$, we define a special type of predicates:
\begin{align}\label{eqn35:High-OrderOpacity} 
  \pred^n_{\mathcal{B}} \subset \Pow_n(Q),
\end{align}
where each level-$(n-1)$ element $Y$ of $\pred^n_{\mathcal{B}}$ satisfies
$\emptyset\not\in Y$ and there is $X\in Y$ and $X'\in \mathcal{B}$ such that
$X'\subset X$.

Change the order-$n$ observer $\Obs(\CC_{A_1\to A_2 \to \cdots \to A_{n-1}\to A_n}^{G,\Obs(G,\Obs(\dots(G,\Obs\overbrace {\scriptstyle)\dots)}^{n-2}})$
to $\Obs(\CC_{A_1\to A_2 \to \cdots \to A_{n-1}\to A_n}^{G,\Obs(G,\Obs(\dots(G,\Det\overbrace {\scriptstyle)\dots)}^{n-2}})$,
that is, replace the concurrent composition $\CC_{A_1\to A_2}^{G,\Obs}$ in
\ref{item9:SEproperty} by $ \CC_{A_1\to A_2}^{G,\Det}$, and keep the 
remaining steps in computing the order-$n$ observer the same. 

The variant order-$n$ observer $\Obs(\CC_{A_1\to A_2 \to \cdots \to A_{n-1}\to A_n}^{G,\Obs(G,\Obs(\dots(G,\Det\overbrace {\scriptstyle)\dots)}^{n-2}})$
can be computed in $(n-1)$-$\EXPTIME$.

Similar to \autoref{thm2:order2SEproperty}, the following result holds.

\begin{theorem}\label{thm2:orderNSEproperty} 
  An FSA $G$ satisfies the order-$n$ estimation-based
  property $\pred^n_{\mathcal{B}}$~\eqref{eqn35:High-OrderOpacity} with respect to agents $A_1,\dots,A_n$ if and only if in every reachable state $\mathcal{X}$ 
  of $\Obs(\CC_{A_1\to A_2 \to \cdots \to A_{n-1}\to A_n}^{G,\Obs(G,\Obs(\dots(G,\Det\overbrace {\scriptstyle)\dots)}^{n-2}})$,
	replace each ordered pair $(\#,\$)$ by $\$$, where
	$\#$ is some state of $G$, denote the most updated $\mathcal{X}$ by $\mathcal{\bar X}$,
	$\mathcal{\bar X}$ belongs to $\pred^n_{\mathcal{B}}$.
\end{theorem}

\autoref{thm2:orderNSEproperty} provides an $(n-1)$-$\EXPTIME$ algorithm for verifying
the order-$n$ estimation-based property $\pred^n_{\mathcal{B}}$~\eqref{eqn35:High-OrderOpacity}.

\subsection{Order-\texorpdfstring{$3$}{3} current-state opacity}

In this subsection, we study a special case of the order-$3$ estimation-based property ---
\emph{order-$3$ current-state opacity}. Recall that the order-$2$ current-state opacity studied in 
\autoref{subsec:Order2Opacity} can be used to describe a scenario ``You don't know what I know''
\cite{Cui2022YouDontKnowWhatIKnow}.
This is particularly useful when a user $\Usr$ wants to operate on a system but an intruder $\Intr$ wants to
attack the system if $\Intr$ knows that $\Usr$ can uniquely determine the current state of the system.
If $\Intr$ cannot know that, then the system is considered to be sufficiently safe and then $\Usr$ will operate 
on the system. However, actually the order-$2$ current-state opacity did not give a complete characterization 
for this scenario, because it has not been guaranteed that $\Usr$ knows $\Intr$ really does not know if $\Usr$ can uniquely
determine the current state. In order to describe this scenario, the order-$3$ current-state opacity is 
necessary: $\Usr$ wants to be sure that $\Intr$ cannot be sure whether $\Usr$ can uniquely determine the current
state --- roughly speaking, ``I know you don't know what I know''. Now $\Usr$ has the defense ability against $\Intr$.

\begin{definition}\label{def1:order-3-opacity} 
  Consider an FSA $G$ as in \eqref{FSA}, user $\Usr$ and intruder $\Intr$ with observable event sets $E_{\Usr}
  \subset E$ and $E_{\Intr}\subset E$, respectively. $G$ satisfies the \emph{order-$3$ current-state opacity
  with respect to $\Usr$, $\Intr$, and $\Usr$} if for all $\alpha\in \ell_{\Usr}(L(G))$, for all $Y\in \Mt_{\Usr\leftarrow
  \Intr\leftarrow\Usr}(G,\alpha)\subset 2^{2^Q}$, $Y\not\subset \{\{q\}|q\in Q\}$.
\end{definition}

Note that \autoref{def1:order-3-opacity} is a special case of \autoref{def:orderNSEproperty},
hence we implicitly assume that $\Usr$ knows $E_{\Intr}$ and $\Intr$ knows $E_{\Usr}$.
Note also that \autoref{def1:order-3-opacity} is a special case of $\pred^3_{\mathcal{B}}$
as in \eqref{eqn35:High-OrderOpacity}, hence can be verified in $2$-$\EXPTIME$ by
\autoref{thm2:orderNSEproperty}.

\begin{example}\label{exam11:High-OrderCSO}
  Consider FSA $G $ as in \autoref{fig15:High-OrderOpacity} studied in 
  \autoref{exam5:High-OrderCSO}, and two agents $\Usr $ and $\Intr $ with their
  observable event sets $E_{\Usr } =\{b,c\}$ and $E_{\Intr }=\{a,b\}$, respectively.
  Assume the labeling functions of $\Usr$ and $\Intr$ are projections.
  The order-$2$ observer $\Obs_{\Usr \leftarrow \Intr }(G )$ is shown in 
  \autoref{fig20:High-OrderOpacity}. The concurrent composition $\CC(G_{\Usr},
  \Obs(\CC_{\Usr  \to  \Intr } ^{G ,\Obs}))$ 
  is shown in \autoref{fig38:High-OrderOpacity}.
  The order-$3$ observer $\Obs_{\Usr \leftarrow \Intr 
  \leftarrow \Usr } (G ) 
  = \Obs(\CC_{\Usr  \to 
  \Intr  \to \Usr } ^{G ,\Obs(G ,\Obs)})
  = \Obs(\CC_{\Usr  \to 
  \Intr  \to \Usr } ^{G ,\Obs(G ,\Det)})$ 
  is shown in \autoref{fig39:High-OrderOpacity}. \autoref{fig40:High-OrderOpacity} is obtained
  from \autoref{fig39:High-OrderOpacity} by replacing each ordered pair $(\#,\$)$ by $\$$, where
  $\#$ is some state of $G$ as in \autoref{thm2:orderNSEproperty}.

  \begin{figure}[!htbp]
  \centering
  \subcaptionbox{The concurrent composition $\CC(G_{\Usr},
  \Obs(\CC_{\Usr  \to  \Intr } ^{G ,\Obs}))$,
    where $G $ is in \autoref{fig15:High-OrderOpacity}, $A=\{(0,\{0,1\}),(2,\{2\})\}$,
    $B=\{(4,\{4,5\}),(5,\{4,5\})\}$, $C=\{(1,\{0,1\})\}$, $D=\{(3,\{3\})\}$.\label{fig38:High-OrderOpacity}}{
  	\begin{tikzpicture}[>=stealth',shorten >=1pt,auto,node distance=2.5 cm, scale = 1.0, transform shape,
	>=stealth,inner sep=2pt]

	\tikzstyle{emptynode}=[inner sep=0,outer sep=0]

	\node[initial, initial where = left, rectangular state] (0) {$(0,A)$};
	\node[rectangular state] (2) [right of =0] {$(2,A)$};
	\node[rectangular state] (4) [right of =2] {$(4,B)$};
	\node[rectangular state] (1) [below =1cm of 0] {$(1,C)$};
	\node[rectangular state] (3) [below =1cm of 2] {$(3,D)$};
	\node[rectangular state] (5) [right of =3] {$(5,B)$};
	
	\node[emptynode] (E1) [left =2.8cm of 0] {};
	\node[emptynode] (E2) [right =2.8cm of 5] {};

	\path [->]
	(0) edge node [above, sloped] {$(c,c)$} (2)
	(2) edge node [above, sloped] {$(b,b)$} (4)
	(0) edge node {$(a,\ep)$} (1)
	(1) edge node [above, sloped] {$(b,b)$} (3)
	(2) edge node [sloped, above] {$(b,b)$} (5)
	;

    \end{tikzpicture}
    }\vspace{0.5cm}

	\subcaptionbox{Order-$3$ observer $\Obs_{\Usr \leftarrow \Intr 
	\leftarrow \Usr } (G )$.\label{fig39:High-OrderOpacity}}{
	  \begin{tikzpicture}[>=stealth',shorten >=1pt,auto,node distance=3.0 cm, scale = 1.0, transform shape,
	>=stealth,inner sep=2pt]

	\node[initial, initial where = above, rectangular state] (0) {$\{(0,A),(1,C)\}$};
	\node[rectangular state] (2) [right of =0] {$\{(2,A)\}$};
	\node[rectangular state] (4) [right of =2] {$\{(4,B),(5,B)\}$};
	\node[rectangular state] (1) [left of = 0] {$\{(3,D)\}$};

	\path [->]
	(0) edge node [above, sloped] {$c$} (2)
	(2) edge node [above, sloped] {$b$} (4)
	(0) edge node [above, sloped] {$b$} (1)
	;

    \end{tikzpicture}
	}\vspace{0.5cm}

	\subcaptionbox{Obtained from order-$3$ observer $\Obs_{\Usr \leftarrow \Intr 
	\leftarrow \Usr } (G )$ by changing each state $\mathcal{X}$
  to $\mathcal{\bar X}$ as in \autoref{thm2:orderNSEproperty}.\label{fig40:High-OrderOpacity}}{
	  \begin{tikzpicture}[>=stealth',shorten >=1pt,auto,node distance=4.0 cm, scale = 1.0, transform shape,
	>=stealth,inner sep=2pt]

	\tikzstyle{emptynode}=[inner sep=0,outer sep=0]

	\node[initial, initial where = above, rectangular state] (0) {$\{\{\{0,1\},\{2\}\},\{\{0,1\}\}\}$};
	\node[rectangular state] (2) [right =1cm of 0] {$\{\{\{0,1\},\{2\}\}\}$};
	\node[rectangular state] (4) [right =1cm of 2] {$\{\{\{4,5\}\}\}$};
	\node[rectangular state] (1) [left =1cm of 0] {$\{\{\{3\}\}\}$};

	
	\path [->]
	(0) edge node [above, sloped] {$c$} (2)
	(2) edge node [above, sloped] {$b$} (4)
	(0) edge node {$b$} (1)
	;

    \end{tikzpicture}
	}\vspace{0.5cm}

	\caption{Illustrative automata in \autoref{exam11:High-OrderCSO}.}
	\end{figure}

	By $\Obs_{\Usr \leftarrow \Intr  \leftarrow \Usr } (G )$ (in \autoref{fig39:High-OrderOpacity}) and \autoref{fig40:High-OrderOpacity}
	we have 
	\begin{subequations}\label{eqn34:High-OrderOpacity} 
	  \begin{align}
		\Mt_{\Usr \leftarrow \Intr  \leftarrow \Usr } (G ,
		\ep) &=  \{\{\{0,1\},\{2\}\},\{\{0,1\}\}\}, \label{eqn34_1:High-OrderOpacity}\\
		\Mt_{\Usr \leftarrow \Intr  \leftarrow \Usr } (G ,
		c) &=  \{\{\{0,1\},\{2\}\}\}, \label{eqn34_2:High-OrderOpacity}\\
		\Mt_{\Usr \leftarrow \Intr  \leftarrow \Usr } (G ,
		cb) &=  \{\{\{4,5\}\}\}, \label{eqn34_3:High-OrderOpacity}\\
		\Mt_{\Usr \leftarrow \Intr  \leftarrow \Usr } (G ,
		b) &=  \{\{\{3\}\}\}.\label{eqn34_4:High-OrderOpacity}
	  \end{align}
	\end{subequations}

	Recall the observer $\Obs_{\Usr }$ of automaton $G_{\Usr}$ shown in
	\autoref{fig18:High-OrderOpacity}. With respect to label sequence $\ep$, $\Usr $'s current-state estimate $\Mt_{\Usr }(G ,\ep)$
	is equal to $\{0,1\}$. By \eqref{eqn34_1:High-OrderOpacity}, $\Usr $
	knows that $\Intr $'s inference of $\Mt_{\Usr }(G ,\ep)$ is either $\{\{0,1\},\{2\}\}$ or $\{\{0,1\}\}$.
	\eqref{eqn34_1:High-OrderOpacity} is computed as follows: When $\Usr $
	observes nothing, the only possible traces are $\ep$ and $a$, then $\Intr $ observes nothing or $a$. By the order-$2$ observer
	$\Obs_{\Usr \leftarrow \Intr }(G )$ (shown in 
	\autoref{fig20:High-OrderOpacity}), when observing nothing, $\Intr $'s
	inference of $\Mt_{\Usr }(G ,\ep)$ is either $\{0,1\}$ or $\{2\}$; when observing $a$, $\Intr $'s
	inference of $\Mt_{\Usr }(G ,\ep)$ is $\{0,1\}$.

	Also by $\Obs_{\Usr }$, $\Mt_{\Usr }(G ,b)=\{3\}$, that is, by observing $b$, $\Usr $ uniquely determines
	the current state of $G $. Then by the order-$3$ observer
	and \eqref{eqn34_4:High-OrderOpacity}, $\Usr $ knows that $\Intr $ exactly knows
	$\Mt_{\Usr }(G ,b)$. Hence system $G $ is not sufficiently safe for
	$\Usr $ to operate on.

\end{example}

\section{Conclusion}

Given a finite-state automaton publicly known to a finite ordered set of agents
$A_1,\dots,A_n$, 
assuming that each agent has its own observable event set of the system and knows
all its preceding agents' observable events, a notion of high-order observer was
formulated to characterize what agent $A_n$ knows about 
what $A_{n-1}$ knows about \dots what $A_2$ knows about $A_1$'s state estimate
of the system. Based on the high-order observer, the state-based properties
studied in labeled finite-state automata have been extended to their high-order versions.
Based on the high-order observer, a lot of further extensions can be done.
For example, in the current paper, only current-state-based properties were considered,
further extensions include initial-state versions, infinite-step
versions, etc.
More importantly, based on the high-order observer, a framework of networked 
labeled finite-state automata can be built in which an agent can infer its upstream agents' 
state estimates, so that all agents can finish a common task based on the network
structure and the agents' inferences to their upstream agents' state estimates.


\section*{Appendix: A comparison between our approach and a special scenario of order-$2$ opacity in \cite{Cui2022YouDontKnowWhatIKnow} and the fundamental mistakes in \cite{Cui2022YouDontKnowWhatIKnow}}
\label{appendix}


In this Appendix we provide a detailed comparison with \cite{Cui2022YouDontKnowWhatIKnow}. In \cite{Cui2022YouDontKnowWhatIKnow}, the authors talked about high-order opacity but they deal with only 2 agents. In detail, the high-order opacity is a special case of our order-2 current-state opacity (\autoref{def5':High-OrderCSO}). Finally, we highlight some issues in the verification methods --- double-observer and state-pair-observer --- proposed in \cite{Cui2022YouDontKnowWhatIKnow}.


In \cite{Cui2022YouDontKnowWhatIKnow}, only projections are considered, which are special labeling
functions as shown before in the paragraph after Eqn.~\eqref{eqn39:High-OrderOpacity}.
Therefore, in this part, we use symbol $P$ to denote a projection instead
of using symbol $\ell$ to denote a general labeling function.
The high-order opacity in \cite{Cui2022YouDontKnowWhatIKnow} is defined as follows.
Consider a set $T_{\spec}\subset\{\{q,q'\}|q, q'\in Q\}$ of state pairs to be distinguished by $\Usr$, i.e., the current-state estimate of $\Usr$ should never contain a pair in $T_{\spec}$ as its subset.
A deterministic FSA $G$ is called \emph{high-order opaque with respect to user $\Usr$, intruder $\Intr$,
and $T_{\spec}$} if
\begin{equation*}
  \tag{A}\label{quoteA:High-OrderOpacity}
  \parbox{\dimexpr\linewidth-4em}{%
	\strut
	for every run $q_0\xrightarrow[]{s}q$ with $q_0\in Q_0$ such that for all $\{q,q'\}\in T_{\spec}$,
	$\{q,q'\}\not\subset \Mt_{\Usr}(G,P_{\Usr}(s))$,
	there is
	another run $q_0'\xrightarrow[]{t}q'$ with $q_0'\in Q_0$ such that 
	$\{q'',q'''\}\subset \Mt_{\Usr}(G,P_{\Usr}(t))$ for some $\{q'',q'''\}\in T_{\spec}$.
	\strut
  }
\end{equation*}

If a deterministic FSA $G$ is high-order opaque with respect to $\Usr$, $\Intr$, and $T_{\spec}$, then
$\Intr$ cannot be sure whether $\Usr$ can distinguish between the states of each state pair of $T_{\spec}$ according 
to $\Usr$'s current-state estimate of $G$.

By definition, high-order opacity of $G$ with respect to $\Usr$, $\Intr$, and $T_{\spec}$, is a special case of
order-$2$ current-state opacity with respect to $\Usr$, $\Intr$, and $\QSordtwo$
as in \autoref{def5':High-OrderCSO}, where $\QSordtwo=\{X\subset Q|(\forall X'\in 
T_{\spec})[X'\not\subset X]\}$,
is also equivalent to, order-$2$ estimation-based property 
\begin{subequations}
  \begin{align}
	&\pred^2(G,\QSordtwo) \\
	= & \{\emptyset\notin Y\subset 2^Q|(\exists X\in Y)(\exists X'\in T_{\spec})[X'\subset X]\}.
  \end{align}
\end{subequations}

\autoref{thm1:order2SEproperty} provides a $2$-$\EXPTIME$ algorithm for verifying high-order opacity of $G$
with respect to $\Usr$, $\Intr$, and $T_{\spec}$. \autoref{thm2:order2SEproperty} provides an $\EXPTIME$ algorithm
for verifying the high-order opacity.

In \cite{Cui2022YouDontKnowWhatIKnow}, an interesting special case of the high-order opacity was also studied:
\begin{equation*}
  \tag{B}\label{quoteB:High-OrderOpacity}
  \parbox{\dimexpr\linewidth-4em}{%
	\strut
	for every run $q_0\xrightarrow[]{s}q$ with $q_0\in Q_0$ and $|\Mt_{\Usr}(G,P_{\Usr}(s))|=1$, there is
	another run $q_0'\xrightarrow[]{t}q'$ with $q_0'\in Q_0$ such that $|\Mt_{\Usr}(G,P_{\Usr}(t))|>1$ and
	$P_{\Intr}(s)=P_{\Intr}(t)$,
	\strut
  }
\end{equation*}
that is, the high-order opacity of deterministic $G$ with 
respect to $\Usr$, $\Intr$, and $T_{\spec}=\{\{q,q'\}|q,q'\in Q,q\ne q'\}$.

In case of \eqref{quoteB:High-OrderOpacity}, whenever the current-state estimate of $\Usr$ is a singleton,
$\Intr$ cannot be sure whether the current-state estimate of $\Usr$ is a singleton or not.

This special case is equivalent to order-$2$ current-state opacity with respect to $\Usr$, $\Intr$, and
$\QSordtwo=\{\{q\}|q\in Q\}$,
is also equivalent to, order-$2$ estimation-based property 
\begin{subequations}
  \begin{align}
	&\pred^2(G,\QSordtwo)\\
	=&\{\emptyset\notin Y\subset 2^Q|(\exists X\in Y)[|X|>1]\}.
  \end{align}
\end{subequations}

In \cite{Cui2022YouDontKnowWhatIKnow}, in order to verify high-order opacity, two methods ---
\emph{double-observer} and \emph{state-pair-observer}, were proposed, where the former runs in 
doubly exponential time and the latter runs in exponential time.
Next, we give counterexamples to show
that neither of them works correctly generally, even with respect to $T_{\spec}=\{\{q,q'\}
|q,q'\in Q,q\ne q'\}$.

The following example shows that both the \emph{double-observer} method and
the \emph{state-pair-observer} method 
fail to verify high-order opacity defined in \eqref{quoteB:High-OrderOpacity}, that is, the main results obtained in 
\cite{Cui2022YouDontKnowWhatIKnow} are generally incorrect.
This example also shows the double-observer generally
cannot correctly compute the inference of the current-state
estimate of $\Usr$ by $\Intr$ defined as $\Mt_{\Usr\leftarrow\Intr}(G,\alpha)$, where $\alpha\in
P_{\Intr}(L(G))$. 

\begin{example}\label{exam5:High-OrderCSO}
  Consider FSA $G_{ \label{cou2}}$ as in \autoref{fig15:High-OrderOpacity}, where $E_{\Usr }
  =\{b,c\}$, $E_{\Intr }=\{a,b\}$.
  We choose $T_{{ }\spec}=\{\{0,1\},\{4,5\}\}$, and show that neither the double-observer nor the state-pair-observer
  can correctly verify the high-order opacity of $G $ (with respect to $\Usr$,
  $\Intr $, and $T_{{ }\spec}$). The variant of observer $\Obs_{\Usr }$ defined in 
  \cite{Cui2022YouDontKnowWhatIKnow}, denoted as $\overline{\Obs_{\Usr }}$, is shown in 
  \autoref{fig16:High-OrderOpacity}. Compared with the standard observer $\Obs_{\Usr }$ as in 
  \autoref{fig18:High-OrderOpacity}, in $\overline{\Obs_{\Usr }}$, there is an additional self-loop on 
  state $\{0,1\}$ with event $a$ because starting from state $0$ there is an unobservable transition with
event $a$ in $G^{ }_{\Usr }$.
  The observer $\Obs_{\Intr }(\overline{\Obs_{\Usr }})$ (i.e., the so-called double-observer defined in 
  \cite{Cui2022YouDontKnowWhatIKnow})
  of $\overline{\Obs_{\Usr }}$ is shown in \autoref{fig17:High-OrderOpacity}.
  In state $\{\{0,1\},\{2\}\}$ of $\Obs_{\Intr }(\overline{\Obs_{\Usr }})$, 
  $(\{0,1\}\times \{0,1\})\cap T_{{ }\spec}
  \ne\emptyset$, in state $\{\{3\},\{4,5\}\}$, $(\{4,5\}\times \{4,5\})\cap T_{{ }\spec}\ne\emptyset$, then by
  \cite[Theorem~1]{Cui2022YouDontKnowWhatIKnow}, $G $ is high-order opaque.

 \begin{figure}[!htbp]
  \centering
  \subcaptionbox{FSA $G $, where $E_{\Usr }=\{b,c\}$, $E_{\Intr }=\{a,b\}$.\label{fig15:High-OrderOpacity}}{
  	\begin{tikzpicture}[>=stealth',shorten >=1pt,auto,node distance=2.5 cm, scale = 1.0, transform shape,
	>=stealth,inner sep=2pt]

	\tikzstyle{emptynode}=[inner sep=0,outer sep=0]

	\node[initial, initial where = left, state] (0) {$0$};
	\node[state] (2) [right of =0] {$2$};
	\node[state] (4) [right of =2] {$4$};
	\node[state] (1) [below =1cm of 0] {$1$};
	\node[state] (3) [below =1cm of 2] {$3$};
	\node[state] (5) [right of =3] {$5$};

	\node[emptynode] (E1) [left of = 0] {};
	\node[emptynode] (E2) [right of = 5] {};

	\path [->]
	(0) edge node [above, sloped] {$c$} (2)
	(2) edge node [above, sloped] {$b$} (4)
	(0) edge node {$a$} (1)
	(1) edge node [above, sloped] {$b$} (3)
	(2) edge node {$b$} (5)
	;

    \end{tikzpicture}
    }

	\subcaptionbox{The variant observer $\overline{\Obs_{\Usr }}$ 
	of automaton $G_{\Usr}$. 
    \label{fig16:High-OrderOpacity}}{
	  \begin{tikzpicture}[>=stealth',shorten >=1pt,auto,node distance=2.5 cm, scale = 1.0, transform shape,
	>=stealth,inner sep=2pt]

	\tikzstyle{emptynode}=[inner sep=0,outer sep=0]

	\node[initial, initial where = above, rectangular state] (A) {$\{0,1\}$};
	\node[rectangular state] (B) [left of = A] {$\{3\}$};
	\node[rectangular state] (C) [right of =A] {$\{2\}$};
	\node[rectangular state] (D) [right of =C] {$\{4,5\}$};

	\node[emptynode] (E1) [left of = B] {};
	\node[emptynode] (E2) [right of = D] {};
	
	\path [->]
	(A) edge node {$b$} (B)
	(C) edge node {$b$} (D)
	(A) edge node {$c$} (C)
	(A) edge [red, loop below] node {$a$} (A)
	;

    \end{tikzpicture}
	}\vspace{0.5cm}

	\subcaptionbox{The observer $\Obs_{\Intr }(\overline{\Obs_{\Usr }})$ 
	of $\overline{\Obs_{\Usr }}$. 
    \label{fig17:High-OrderOpacity}}{
	  \begin{tikzpicture}[>=stealth',shorten >=1pt,auto,node distance=3.5 cm, scale = 1.0, transform shape,
	>=stealth,inner sep=2pt]

	\tikzstyle{emptynode}=[inner sep=0,outer sep=0]

	\node[initial, initial where = left, rectangular state] (AC) {$\{\{0,1\},\{2\}\}$};
	\node[rectangular state] (BD) [right of =AC] {$\{\{3\},\{4,5\}\}$};

	\node[emptynode] (E1) [left of = AC] {};
	\node[emptynode] (E2) [right of = BD] {};

	\path [->]
	(AC) edge node [above, sloped] {$b$} (BD)
	(AC) edge [loop above] node {$\red a$} (AC)
	;

    \end{tikzpicture}
	}\vspace{0.5cm}

	\subcaptionbox{The state-pair-observer of $G $ with respect to $\Usr $ and $\Intr $.
	\label{fig21:High-OrderOpacity}}{
	  \begin{tikzpicture}[>=stealth',shorten >=1pt,auto,node distance=4.5 cm, scale = 1.0, transform shape,
	>=stealth,inner sep=2pt]

	\tikzstyle{emptynode}=[inner sep=0,outer sep=0]

	\node[initial, initial where = left, rectangular state] (A) {$\left\{ 
	\begin{matrix}
	  (0,0),(1,0),(0,1)\\
	  (1,1),(2,2)
	\end{matrix}\right\}$};
	\node[rectangular state] (B) [right of =A] {$\left\{ 
		\begin{matrix}
		  (0,0),(0,1)\\
		  (1,0),(1,1)
		\end{matrix}
	\right\}$};
	\node[rectangular state] (C) [below =1cm of A] {$\left\{ 
		\begin{matrix}
		  (4,4),(4,5),(5,4)\\
		  (5,5),(3,3)
		\end{matrix}
	\right\}$};

	\node[emptynode] (E1) [left of = A] {};
	\node[emptynode] (E2) [right of = B] {};

	\path [->]
	(A) edge node [above, sloped] {$a$} (B)
	(A) edge node {$b$} (C)
	(B) edge node [above, sloped] {$b$} (C)
	(B) edge [loop above] node {$a$} (B)
	;

    \end{tikzpicture}
	}\vspace{0.5cm}

	\caption{Illustration of methods in \cite{Cui2022YouDontKnowWhatIKnow}.}
	\end{figure}

  The state-pair-observer of $G $ with respect to $\Usr $ and $\Intr $ 
  shown in \autoref{fig21:High-OrderOpacity} is the observer of the product
  of $G$ and itself with respect to $\Intr$, where the two trajectories in the product 
  look identifical to $\Usr$. In \autoref{fig21:High-OrderOpacity},
  no state
  has empty intersection with $T_{{ }\spec}$. Then by \cite[Theorem~2]{Cui2022YouDontKnowWhatIKnow},
  one concludes that $G $ is high-order opaque.

  Directly by definition, from $ab\in P_{\Intr }(L(G ))$, we have 
  $\Mt_{\Usr \leftarrow\Intr }(G ,ab)=
  \{\{3\}\}$. We also have $(\{3\}\times\{3\})\cap T_{{ }\spec}=\emptyset$, then $G $ is not
  high-order opaque. However,
  in double-observer $\Obs_{\Intr }
  (\overline{\Obs_{\Usr }})$, $\Mt_{\Usr \leftarrow\Intr }(G ,ab)$
  is wrongly computed as $\{\{3\},\{4,5\}\}$.

  Next, we use our method to show that $G $ is not high-order opaque.
  The observer $\Obs_{\Usr }$ is shown in \autoref{fig18:High-OrderOpacity}. The concurrent composition
  $\CC(G_{\Usr},\Obs_{\Usr})$ is shown in \autoref{fig19:High-OrderOpacity}.
  The order-$2$ observer $\Obs_{\Usr \leftarrow\Intr }(G )$ 
  is shown in \autoref{fig20:High-OrderOpacity}.
  Consider state $\{(3,\{3\})\}$ of $\Obs_{\Usr \leftarrow\Intr }(G )$,
  $(\{3\}\times \{3\})\cap T_{{ }\spec}=\emptyset$, by \autoref{thm1:order2SEproperty},
  $G $ is not high-order opaque.

  \begin{figure}[!htbp]
  \centering
  \subcaptionbox{The observer $\Obs_{\Usr }$ of automaton $G^{ }_{\Usr }$ (shown in \autoref{fig15:High-OrderOpacity}).\label{fig18:High-OrderOpacity}}{
  \begin{tikzpicture}[>=stealth',shorten >=1pt,auto,node distance=2.5 cm, scale = 1.0, transform shape,
	>=stealth,inner sep=2pt]

	\tikzstyle{emptynode}=[inner sep=0,outer sep=0]

	\node[initial, initial where = above, rectangular state] (A) {$\{0,1\}$};
	\node[rectangular state] (B) [left of = A] {$\{3\}$};
	\node[rectangular state] (C) [right of =A] {$\{2\}$};
	\node[rectangular state] (D) [right of =C] {$\{4,5\}$};

	\path [->]
	(A) edge node {$b$} (B)
	(C) edge node {$b$} (D)
	(A) edge node {$c$} (C)
	;

    \end{tikzpicture}
  }\vspace{0.5cm}

  \subcaptionbox{The concurrent composition $\CC(G_{\Usr},\Obs_{\Usr})$.\label{fig19:High-OrderOpacity}}{
	\begin{tikzpicture}[>=stealth',shorten >=1pt,auto,node distance=3.5 cm, scale = 1.0, transform shape,
	>=stealth,inner sep=2pt]

	\tikzstyle{emptynode}=[inner sep=0,outer sep=0]

	\node[initial, initial where = left, rectangular state] (0A) {$(0,\{0,1\})$};
	\node[rectangular state] (1A) [below =1cm of 0A] {$(1,\{0,1\})$};
	\node[rectangular state] (2C) [right of =0A] {$(2,\{2\})$};
	\node[rectangular state] (4D) [right of =2C] {$(4,\{4,5\})$};
	\node[rectangular state] (3B) [right of =1A] {$(3,\{3\})$};
	\node[rectangular state] (5D) [right of =3B] {$(5,\{4,5\})$};

	\node[emptynode] (E1) [left of = 0A] {};
	\node[emptynode] (E2) [right of = 5D] {};
	
	\path [->]
	(0A) edge node {$(a,\ep)$} (1A)
	(2C) edge node {$(b,b)$} (4D)
	(0A) edge node {$(c,c)$} (2C)
	(1A) edge node {$(b,b)$} (3B)
	(2C) edge node [sloped, above] {$(b,b)$} (5D)
	;

    \end{tikzpicture}
  }\vspace{0.5cm}

  \subcaptionbox{The order-$2$ observer $\Obs_{\Usr \leftarrow\Intr }(G )$.\label{fig20:High-OrderOpacity}}{
	\begin{tikzpicture}[>=stealth',shorten >=1pt,auto,node distance=5.5 cm, scale = 1.0, transform shape,
	>=stealth,inner sep=2pt]

	\node[initial, initial where = left, rectangular state] (0A2C) {$\{(0,\{0,1\}),(2,\{2\})\}$};
	\node[rectangular state] (4D5D) [right of =0A2C] {$\{(4,\{4,5\}),(5,\{4,5\})\}$};
	\node[rectangular state] (3B) [below =1cm of 4D5D] {$\{(3,\{3\})\}$};
	\node[rectangular state] (1A) [left of =3B] {$\{(1,\{0,1\})\}$};

	\path [->]
	(0A2C) edge node [above, sloped] {$b$} (4D5D)
	(0A2C) edge node {$a$} (1A)
	(1A) edge node [above, sloped] {$b$} (3B)
	;

    \end{tikzpicture}
  }\vspace{0.5cm}

	\caption{Illustration of methods in the current paper.}
	
	\end{figure}

\end{example}


\end{document}